\documentclass[journal, 10pt]{IEEEtran}

\usepackage{eulervm}

\usepackage[english]{babel}
\usepackage{amsmath,amssymb}
\usepackage{times}
\usepackage{relsize}
\usepackage{graphicx}
\usepackage{url}
\usepackage{booktabs}
\usepackage{multirow}
\usepackage{mparhack}
\usepackage{subfigure}
\usepackage{authblk}
\usepackage{amsthm}

\usepackage[noend]{algorithmic}
\usepackage[linesnumbered,ruled]{algorithm2e}

\usepackage{array}
\usepackage{cite}
\usepackage{eqparbox}
\usepackage{mdwmath}

\usepackage{epsfig}
\usepackage{xcolor}

\usepackage{mathtools}
\usepackage{bm}
\usepackage{bbold}
\usepackage[all]{xy}

\usepackage{accents}

\newcommand\overplus[1]{\accentset{+}{#1}}
\newcommand\overminus[1]{\accentset{-}{#1}}

\newcommand\overcirc[1]{\accentset{\circ}{#1}}
\newcommand\undercirc[1]{\underaccent{\circ}{#1}}

\newtheorem{theorem}{Theorem}
\newtheorem{definition}{Definition}
\newtheorem{lemma}[theorem]{Lemma}

\newtheorem{corollary}[theorem]{Corollary}
\newtheorem{property}{Property}

\DeclareFontFamily{OT1}{pzc}{}
\DeclareFontShape{OT1}{pzc}{m}{it}%
  {<-> s * [1.1] pzcmi7t}{}
\DeclareMathAlphabet{\mathpzc}{OT1}{pzc}%
                     {m}{it}

\def\I{\mathcal{I}}
\def\J{\mathcal{J}}
\def\L{\mathcal{L}}

\newcommand{\MSL}{\textit{MinSumL}}
\newcommand{\MML}{\textit{MinMaxL}}

\title{Load Optimization with User Association in 
Cooperative and Load-Coupled LTE Networks}

\author{Lei You}
\author{Di Yuan}
\affil[1]{{\small Department of Science and Technology, Link{\"o}ping
University, Sweden}}
\affil[ ]{\texttt{{\small \{lei.you; di.yuan@it.uu.se\}}}\/}

\begin{document}

\maketitle

\begin{abstract} 
We extend the problem of optimizing user association for load balancing in
cellular networks along two dimensions. First, we consider joint transmission
(JT), which is one of the coordinated multipoint (CoMP) techniques, with which a
user may be simultaneously served by multiple base stations.  Second, we account
for, mathematically, the coupling relation between the base stations' load
levels that are dependent on each other due to
inter-cell interference.  We formulate two optimization problems, sum load
minimization (\MSL{}) and maximum load minimization (\MML{}). We prove that both
\MSL{} and \MML{} are $\mathcal{NP}$-hard.
We propose a mixed integer linear programming (MILP) based scheme by means of linearization.
This approach also leads to a bounding scheme for performance benchmarking. 
Then, we derive a set of partial optimality
conditions. Fulfillment of the conditions will guarantee performance
improvement for both \MSL{} and \MML{}. A solution algorithm is then
derived based on the conditions. Simulation results are provided to
demonstrate the effectiveness of the approaches.
\end{abstract}


\section{Introduction}

\subsection{Background}

To provide high rate and capacity in the fifth generation (5G) communication
systems, efficient utilization of time-frequency resource is crucial. The
efficiency of time-frequency resource units (RUs) is influenced by many aspects,
including radio access, radio resource management, network planning,
interference coordination/mitigation, and user association strategies.
Inappropriate user association may result in inefficient usage of RUs.  In
particular, the user association policy is crucial in terms of balancing the
resource usage among cells and avoid overloading, especially in heterogeneous
networks (HetNets) scenarios due to the large difference between the transmit
power levels of macro cells (MCs) and small cells (SCs).

In this paper, we extend the optimization problem of user association along two
dimensions. The first is the consideration of coordinated multipoint (CoMP)
transmission \cite{Porcello:2014wq}. As one of the CoMP techniques, Joint
Transmission (JT) allows a user equipment (UE) to simultaneously receive its data from multiple
base stations (BSs) on the same set of time-frequency resource. Second, we
account for mathematically the load-coupling relation among the cells. Here, the
load refers to the amount of utilized RUs, and this
concept has been very useful for performance characterization
\cite{Siomina:2009bp,Cavalcante:2014jd2,IViering:2009tq,Fehske:2013gn,Siomina:eq,
Fehske:2012iw,Cavalcante:2014jd,You:2015wva,Siomina:2013ew,Siomina:2014be}.
The load in one cell governs the amount of transmission
and consequently inter-cell interference.  Thus, even for given association
of the UEs, the cell load levels cannot be treated 
independently from each
other.  Rather, they are inherently coupled, and constitute the so
called \textit{load-coupling system} that mathematically formulates the
mutual influence. To our knowledge, this type of characterization has not been
developed for JT\@.

\subsection{Related work}

User association optimization subject to quality of service (QoS) constraints
has been studied in~\cite{Athanasiou:2015hb, Sun:2015kn, Anonymous:M4eIBsQg,
Singh:2014bs, Li:2013ho, Hong:2012ky} for non-JT scenarios. In
\cite{Athanasiou:2015hb}, the authors applied Lagrangian duality for 
min-max fairness. Under the same performance
objective, the authors of~\cite{Sun:2015kn} studied BS association and power
allocation. In~\cite{Anonymous:M4eIBsQg},
optimal user association for ultra-dense networks has been approached by means
of integer linear programming.  User association with the perspective of random
HetNets, where access points are distributed according to a stationary point
process, has been studied in~\cite{Singh:2014bs}. It is shown that, if the
association is based on received power, the amount of association to SCs
decreases with the path loss exponent and increases with the channel gain
variance. In~\cite{Li:2013ho}, intracell cooperation using relays for resource
allocation optimization is investigated.  By using the KKT conditions, the
problem is decomposed into several, independent sub-problems.  User association
and resource allocation with the presence of possible misreporting of
information from individual users have been addressed by game theory in
\cite{Hong:2012ky}.

Radio resource allocation for CoMP has been studied in~\cite{Ye:2013ha,
Nigam:2014cd, Lakshmana:2014eu}. In~\cite{Ye:2013ha}, the authors considered
optimizing user association for load balancing, in which the load refers 
to the number of users associated to the cell, and thus load balancing
implicitly assumes that all users are uniform in resource consumption.
The study in~\cite{Nigam:2014cd} applied stochastic geometry for analyzing
network coverage with cooperative transmission.
Resource allocation in non-coherent JT-CoMP
scenarios has been studied in~\cite{Lakshmana:2014eu}, by optimizing a
fractional frequency reuse scheme.

\subsection{Our Work}

There is a lack of investigation that accounts for the
load-coupling relation with JT in the context of user association.  The
load-coupling model, in which the load represents the amount of time-frequency
resource consumption, rather than counting the number of users, enables an
analytical characterization of the user-specific resource requirement, and
the dynamic dependence relation of resource consumption among cells due to
mutual interference
\cite{Siomina:eq,Fehske:2012iw,Cavalcante:2014jd,Siomina:2014be}.

The load-coupling model consists in a group of non-linear equations.  Accounting
for JT adds to the complexity of load-coupling analysis.
We study optimizing
the cell-UE association with JT and cell-load-coupling. We consider two
performance metrics. The first is the total amount of time-frequency resource
consumption, i.e., sum of cell load levels. The second metric is the
minimization of maximum load, for the effect of load balancing among cells.

Our main contributions are as follows. We extend and generalize the
load-coupling model proposed in~\cite{Siomina:2009bp} to 
JT, and formulate both the sum load minimization (\MSL{}) and the
maximum load minimization (\MML{}) with the new model. Both \MSL{} and
\MML{} are proved to be $\mathcal{NP}$-hard. We first provide a solution approach
based on a linearization of the non-linear load coupling
constraint. The two problems are then approximately represented by two MILPs.
The approach also yields a bounding scheme
for performance benchmarking.
In our second approach, we derive a set of partial
optimality conditions; fulfillment of the conditions will guarantee performance
improvement for both \MSL{} and \MML{}. These conditions provide a theoretical
basis for deriving a solution algorithm referred to as \textit{Minimization of Load}
(\textit{MinL}), for optimizing the association between cells and UEs for
\MSL{} and \MML{}. We present extensive results to demonstrate the
effectiveness of the solution approaches as well as the benefit of JT in load
optimization.

\section{System Model and Load Coupling} 
\label{sec:system_model}

\subsection{Notations} 

We now generalize the model in  
\cite{Siomina:2009bp,Cavalcante:2014jd2,IViering:2009tq,Fehske:2013gn,Siomina:eq,
Fehske:2012iw,Cavalcante:2014jd,Siomina:2013ew,Siomina:2014be} to JT.
Denote the set of cells by $\mathcal{I}$ and the set of UEs by $\mathcal{J}$.
Let $n=|\mathcal{I}|$ and $m=|\mathcal{J}|$. Denote by ${\mathcal{I}^{+}_j}$ the
set of candidate serving cells of UE $j$.  For any UE $j$, denote by
$\mathcal{J}^{+}_i$ the set of UEs that potentially can be served by cell $i$.
Apparently, $i\in\mathcal{I}^{+}_j\Leftrightarrow
j\in\mathcal{J}^{+}_i$. Each UE $j$ is assigned with a home cell $c_j$, which is
chosen as the cell with the best received power at UE $j$. The union of all UEs'
home cells is denoted by $\mathcal{C}$. Let $\mathcal{J}^{-}_i=\{j:c_j=i\}$,
i.e., the set of UEs of which the corresponding home cell is $i$. 

We allow JT in our model, so any UE $j$ can be served by multiple
cells simultaneously. For the sake of presenting the system model,
consider a (generic) cell-UE association solution, for which the set
of cells serving UE $j$ and the set of UEs served by cell $i$ are
denoted by $\mathcal{I}_j$ and $\mathcal{J}_i$, respectively.  Note
that $i\in\mathcal{I}_j\Leftrightarrow j\in\mathcal{J}_i$. For the
sake of presentation, we selectively use either according to the
context.  Note that we have
$\mathcal{I}_j\subseteq{\mathcal{I}^{+}_j}$ and
$\mathcal{J}_i\subseteq{\mathcal{J}^{+}_i}$ for any $j\in\mathcal{J}$
and $i\in\mathcal{I}$, respectively. For each UE $j$, we use the
notation $d_j$ to denote its bit rate demand. 
For any $j\in\mathcal{J}$, we assume
$c_j\in\mathcal{I}_j$. Thus $\mathcal{J}^{-}_i\subseteq\mathcal{J}_i$
holds for any $i\in\mathcal{I}$.

\subsection{Load Coupling} 

For this moment, we fix the cell-UE association. Without loss of
generality, we use RU as the minimum unit for resource allocation, composed by
one or more resource blocks in orthogonal frequency division multiple access
(OFDMA). 

\vspace{-3mm}
\begin{equation} 
\gamma_{j}\triangleq{h}_j(\bm{x})= \frac{%
\sum_{i\in \mathcal{I}_j}p_{i}g_{ij} }{%
\sum_{k\in
\mathcal{I} \backslash \mathcal{I}_j}p_{k}g_{kj}x_k+\sigma^2 
} 
\label{eq:sinr} 
\end{equation}

\vspace{-3mm}
\begin{equation} 
\bm{\gamma}\triangleq[\gamma_1,\gamma_2,\ldots,\gamma_m] 
\label{eq:network_sinr} 
\end{equation}

Eq.~(\ref{eq:sinr}) models the SINR at UE $j\in\mathcal{J}$. The network-wide
SINRs are represented by the vector shown in Eq.~(\ref{eq:network_sinr}). In
Eq.~(\ref{eq:sinr}), $p_i$ ($p_i>0$) denotes the transmit power per RU (in time
and frequency) of cell $i$. Notation $g_{ij}$ is the power gain between cell $i$
and UE $j$. The summation $\sum_{i\in\mathcal{I}_j}p_{i}g_{ij}$ is the received
signal power at UE $j$ via JT, from all the UE $j$'s serving cells in
$\mathcal{I}_j$. In the denominator, $\sigma^2$ is the noise power. Entity $x_k$
is the load of cell $k$, which is defined to be the proportion of RUs consumed
in cell $k$ for all UEs in $\mathcal{J}_k$. We model the interference that UE
$j$ receives from other cells by the term $\sum_{k\in \mathcal{I} \backslash
\mathcal{I}_j}p_{k}g_{kj}x_k$. For any RU in cell $i$, $x_k$ is
intuitively interpreted as the likelihood that the served UEs of cell
$i$ receive the interference from the cell $k$.  The SINR of any UE
$j$ is a function of the load vector $\bm{x}$, denoted by
$h_j(\bm{x})$.

\vspace{-3mm}
\begin{equation} 
x_i\triangleq{f}_i(\bm{\gamma})=\sum_{j\in\mathcal{J}_i}\frac{d_j}{MB\log
_2\left(1+\gamma_j\right)}
\label{eq:load} 
\end{equation}

\vspace{-3mm}
\begin{equation} 
\bm{x}\triangleq[x_1,x_2,\ldots,x_n] 
\label{eq:network_load} 
\end{equation}

The load of any cell $i$, is represented in Eq.~(\ref{eq:load}), so as to satisfy
the bit rate demands of its served UEs. Eq.~(\ref{eq:network_load}) gives the
network-wide cell load.  In the denominator, $B$ is the bandwidth per RU and $M$
is the total number of RUs available. The entity $B\log_2(1+\gamma_j)$ is the
achievable bit rate per RU\@. Thus, $MB\log_2(1+\gamma_j)$ computes the total
achievable bit rate for UE $j$.  The term $d_j/MB\log_2(1+\gamma_j)$ computes
the required amount of RUs to satisfy the demand $d_j$. Then
$x_i=\sum_{j\in\mathcal{J}_i}d_j/MB\log_2(1+\gamma_j)$ is the proportion of the
RUs consumed for transmission in cell $i$, which complies to the definition of
the cell load. It can be verified in Eq.~(\ref{eq:load}), that if any UE $j$
satisfies $j\in\mathcal{J}_i\cap\mathcal{J}_k~(k\neq i)$ (meaning that $j$ is
currently served by cell $i$ and $k$ via JT), then the number of RUs
consumed by UE $j$ in both cells $i$ and $k$ are equal. The load of cell $i$ is
a function of the SINRs $\gamma_j$ for all $j\in\mathcal{J}_i$, denoted by
$f_i(\bm{\gamma})$ in Eq.~(\ref{eq:load}).

We have the coupling equations for both UE's SINR and cell's load, shown in
Eq.~(\ref{eq:load_coupling}). 

\vspace{-3mm}
\begin{equation} 
\textnormal{load-coupling: }\left\{
\begin{array}{l} \bm{\gamma}=\bm{h}(\bm{x}) \\
\bm{x}=\bm{f}(\bm{\gamma}) 
\end{array}\right.  
\label{eq:load_coupling} 
\end{equation}

Lemma~\ref{lma:gamma_x} shows that a larger $\bm{\gamma}$ leads to a smaller
$\bm{x}$, and vice versa. This can be verified easily in Eq.~(\ref{eq:sinr}) and
Eq.~(\ref{eq:load}).

\begin{lemma} The following relationships hold for the coupling
equations~\eqref{eq:load_coupling}.
\begin{enumerate} 
\item $\bm{{f}}(\bm{\gamma})\leq\bm{{f}}(\bm{\gamma}')$, for
$\bm{\gamma}\geq\bm{\gamma}'$.
\item $\bm{{h}}(\bm{x})\leq\bm{{h}}(\bm{x}')$, for $\bm{x}\geq\bm{x}'$.  
\end{enumerate}
\label{lma:gamma_x} 
\end{lemma}

We remark that the way of modeling interference above is not exact. 
For two base stations serving a UE with JT,
the interference generated to another UE should contain
an additional term\footnote{For cells $i_1$ and $i_2$ serving jointly
a UE, the extra term for the interference to UE $j$  
equals $2\sqrt{p_{i_1}} \sqrt{p_{i_2}} h_{i_1 j} h_{i_2, j} 
\cos(\theta_{i_1} - \theta_{i_2})$, 
where $h_{i_1 j}$ and $h_{i_2, j}$ are the channel gains, and
$\theta_{i_1}$ and $\theta_{i_2}$ are the phases of the two received
interfering signals, see \cite{BaGi15}.}. One reason of using the
approximation is to achieve a good trade-off between exactness and
complexity, as an exact interference modeling would require the
network to acquire and process, for each UE, information related to
other individual UEs (rather than information available at the cell
level). Moreover, the extra terms incurred due to multiple JT
operations tend to offset each other (for a theoretical proof, see
\cite{BaGi15}). In addition, 
the approximation is coherent with the load-coupling modeling approach.
The approach itself is an approximation of interference (even for the non-JT
case), but it suits well (see
\cite{Siomina:2009bp,Cavalcante:2014jd2,IViering:2009tq,Fehske:2013gn,Siomina:eq,
Fehske:2012iw,Cavalcante:2014jd,Siomina:2013ew,Siomina:2014be}) as long as the
performance of interest is at an
aggregated level with a time scale being greater than that of an RU,
which is the case in our study. For these reasons, our way of modeling
interference has been used by a number of other authors for JT (e.g.,
\cite{Nigam:2014cd,TaSiAnJo14,BaGi15}).

\subsection{Standard Interference Function}

We show that the load-coupling equations can be solved by
the fixed-point iteration, based on the fact that both
$\bm{{h}}(\bm{{f}}(\bm{\gamma}))$ and $\bm{{f}}(\bm{{h}}(\bm{x}))$ are standard
interference function (SIF).

\vspace{-2mm}
\begin{definition} 
A function $\bm{\vartheta}$: $\mathbb{R}^m_+\rightarrow\mathbb{R}_{++}$ is
called an SIF if the
following properties hold: 
\begin{enumerate} 
\item (Scalability)
$\alpha\bm{\vartheta}(\bm{\mu})>\bm{\vartheta}(\alpha\bm{\mu}),~
\bm{\mu}\in \mathbb{R}^m_+,~\alpha>1$.
\item (Monotonicity) $\bm{\vartheta}(\bm{\mu})\geq \bm{\vartheta}(\bm{\mu}')$,
if $\bm{\mu} \geq \bm{\mu}'$.
\end{enumerate} 
\label{def:SIF} 
\end{definition}

\begin{property} 
Suppose the function $\bm{\vartheta}$ is an SIF\@. For the sequence
$\bm{\mu}^{(0)},\bm{\mu}^{(1)},\ldots$ generated by fixed-point iterations, if
there exists $k$ satisfying
$\bm{\vartheta}(\bm{\mu}^{(k)}) \leq\bm{\vartheta}(\bm{\mu}^{(k-1)})$, then the
sequence
$\bm{\mu}^{(k)},\bm{\mu}^{(k+1)},\ldots$ is monotonously decreasing (in every
component), and converges to a unique
fixed point.  
\end{property}

\begin{property} 
Suppose $\bm{A}\in\mathbb{R}^{n\times m}$, $\bm{\mu}\in\mathbb{R}^{m}$ and
$\bm{b}\in\mathbb{R}^{n}$. Define
$\varphi(\bm{\mu}):\mathbb{R}^n_+\rightarrow\mathbb{R}$. If $\varphi_1$ is
concave in $\bm{\mu}$, so is $\varphi (\bm{A}\bm{\mu}+\bm{b})$.  
\end{property}

\begin{property} 
Denote $\varphi_1:\mathbb{R}^k\rightarrow\mathbb{R}$. Denote
$\varphi_2:\mathbb{R}^n\rightarrow\mathbb{R}^k$. Then
the function $\varphi_1(\varphi_2(\cdot))$ is concave, if $\varphi_1$ is
concave and nondecreasing, and $\varphi_2$
is concave.  
\end{property}

\begin{lemma} 
The two functions \!$\bm{f}(\bm{h}(\bm{x}))$ and $\bm{h}(\bm{f}(\bm{\gamma}))$
are SIF\@.
\label{lma:sif} 
\end{lemma}
\begin{proof} 
The monotonicity is shown by Lemma~\ref{lma:gamma_x}. To prove the scalability,
we first show the concavity of $\bm{h}(\bm{f}(\bm{\gamma}))$ in $\bm{\gamma}$.
Define function $\zeta$ as 
$\zeta(\bm{\varphi})\triangleq{\sum_{i\in
    \mathcal{I}_j}p_{i}g_{ij}}/{\left(\sum_{k\in \mathcal{I} \backslash
    \mathcal{I}_j}p_{k}g_{kj}\sum_{j\in\mathcal{J}_i}\frac{r_j}{\varphi_j}+\sigma
^2\right)} $.
Note that function $\mathbb{R}\rightarrow\mathbb{R}:x\mapsto1/(\frac{1}{x}+1)$
is concave.  Combined with Property~2, $\zeta$ is concave. Define the function
$\varphi_j$ as 
$
\varphi_j(\gamma_j)\triangleq\log(1+\gamma_j) 
$. 
Note that function $\mathbb{R}\rightarrow\mathbb{R}:x\mapsto\log_2(1+x)$ is
concave. According to Property~2, $\varphi_j$ is concave in $\gamma_j$. We
remark that $\zeta$ is concave and nondecreasing. By Property 3,
$h_j=\zeta(\bm{\varphi}(\bm{\gamma}))$ is concave in $\bm{\gamma}$.  We then
show $\bm{f}(\bm{h}(\bm{x}))$ is concave in $\bm{x}$. Note that function
$\mathbb{R}\rightarrow\mathbb{R}:x\mapsto1\log[1/(1+\frac{1}{x})]$ is concave,
thus for any $i\in[1,n],~f_i(\bm{h}(\bm{x}))$ is concave, according to
Property~2. By the conclusion in~\cite{Fehske:2012iw} that the scalability holds
for any strictly concave function. Hence the conclusion.
\end{proof}

By Lemma~\ref{lma:sif} and Property~1, there exists a unique fixed
point for function $\bm{f}(\bm{h}(\bm{x}))$. The same conclusion
applies for function
$\bm{h}(\bm{f}(\bm{\gamma}))$. 

\begin{lemma}
Suppose $\tilde{\bm{x}}$ is the fixed point of function
$\bm{f}(\bm{h}(\bm{x}))$, i.e., $\tilde{\bm{x}}$ satisfies
$\tilde{\bm{x}}=\bm{f}(\bm{h}(\tilde{\bm{x}}))$, and 
let $\tilde{\bm{\gamma}}=\bm{h}(\tilde{\bm{x}})$. Then $\tilde{\bm{\gamma}}$
is the fixed point of function $\bm{h}(\bm{f}(\bm{\gamma}))$, i.e.,
$\tilde{\bm{\gamma}}=\bm{h}(\bm{f}(\tilde{\bm{\gamma}}))$ holds for $\tilde{\bm{\gamma}}$,
 and vice versa.
\label{lma:vice_versa}
\end{lemma}
\begin{proof}
Since $\tilde{\bm{x}}=\bm{f}(\tilde{\bm{\gamma}})$, by
$\tilde{\bm{\gamma}}=\bm{h}(\tilde{\bm{x}})$, we get
$\tilde{\bm{\gamma}}=\bm{h}(\bm{f}(\tilde{\bm{\gamma}}))$. 
\end{proof}

\section{Problem Formulation and Complexity Analysis} 
\label{sec:formulation}

\subsection{Formulation}

The two problems \MSL{} and \MML{} are formulated in this section. Here,
the association is represented by an $n\times m$
matrix $\bm{\kappa}$, and $\kappa_{ij}=1$ if and only if cell $i$ is currently
serving UE $j$.

\vspace{-5mm}
\begin{subequations} 
\begin{alignat}{2} 
[\MSL{}]&\quad\min\limits_{\bm{\kappa}} \quad \sum_{i\in\mathcal{I}}x_i \\
\textnormal{s.t.} & \quad \bm{x}=\bm{f}(\bm{h}(\bm{x},\bm{\kappa}),\bm{\kappa})
\\ 
&\quad 0< x_{i}\leq 1
\quad\quad\quad\quad\quad~
 i\in\mathcal{I} \\ 
& \quad \kappa_{ij}=0
\quad\quad~~~~\!
 i\notin\mathcal{I}^{+}_j, j\in\mathcal{J}
\\ 
& \quad \kappa_{ij}=1
\quad\quad~~~~\!
 i\in\mathcal{C},  j\in\mathcal{J}^{-}_i \\ 
& \quad \kappa_{ij}\in\{0,1\}
~~~~
 i\in\mathcal{I},~j\in\mathcal{J} 
\end{alignat} 
\label{eq:p0}
\end{subequations} 

\vspace{-8mm}
\begin{subequations} 
\begin{alignat}{2} 
[\MML{}]&\quad\min\limits_{\bm{\kappa}} \quad \max\limits_{i\in\mathcal{I}}~x_i
\\ 
\textnormal{s.t.} & \quad \textnormal{Eq.~(\ref{eq:p0}b) -- Eq.~(\ref{eq:p0}f)}
\end{alignat} 
\label{eq:p1}
\end{subequations}

The optimization variable in both two problems is $\bm{\kappa}$. In \MSL{},
load-coupling constraints are shown in Eq.~(\ref{eq:p0}b). Eq.~(\ref{eq:p0}c)
guarantees that the cell load value is positive but less or equal than 1. The
variable $\bm{\kappa}$ is binary.  Eq.~(\ref{eq:p0}d) imposes that
$\mathcal{I}_j\subseteq{\mathcal{I}^{+}_j}$ and
$\mathcal{J}_i\subseteq{\mathcal{J}^{+}_i},~ i\in\mathcal{I}$.
Constraint Eq.~(\ref{eq:p0}e) imposes
$\mathcal{J}_i^{-}\subseteq\mathcal{J}_i,~ i\in\mathcal{I}$. \MML{}
differs only in the objective with \MSL{}.

\subsection{Complexity Analysis}

\begin{figure*}
\resizebox{1.0\linewidth}{!}{
\xymatrix{c_0\ar[d]^{n+1}& & {c_1} \ar@{>}[d]^{0.5} & & &
{c_2}\ar@{>}[d]^{0.5} & & \ldots & & {c_n} \ar@{>}[d]^{0.5} & \\u_0&
a_1\ar@{.>}[l]|-{1}\ar@{~>}[r]^{0.5}\ar@{.>}@/_1pc/[rrrd]|-{1} & u_1 &
a'_1\ar@{.>}@/_1pc/[rrrrrd]|-{1}\ar@{.>}@/^0.8pc/[lll]|-{1}\ar@{~>}[l]_{0.5}
&
a_2\ar@{.>}@/_1pc/[llll]|-{1}\ar@{.>}[d]|-{1}\ar@{.>}@/_/[rrrrd]|-{1}
\ar@{~>}[r]^{0.5} & u_2 & a'_2
\ar@{.>}@/_1.5pc/[llllll]|-{1}\ar@{~>}[l]_{0.5} &\ldots & a_n
\ar@{.>}@/_2pc/[llllllll]|-{1}\ar@{~>}[r]^{0.5}\ar@{.>}@/_1pc/[lllld]|-{1}
& u_n & a'_n \ar@{.>}@/_4pc/[llllllllll]|-{1}\ar@{~>}[l]_{0.5}
\ar@{.>}@/^1pc/[lld]|-{1}\\ & & &{c_{n+1}}\ar@{>}[r]_{3} &
\underset{(x_1\vee{x_2}\vee{x_n})}{u_{n+1}} & \ldots & &
{c_{n+m}}\ar@{>}[r]_{3} & \underset{(\hat{x}_1\vee{x_2}
\vee{\hat{x}_n})}{u_{n+m}} & } }
\caption{The instance for $\mathcal{NP}$-hardness reduction. 
There are $n+m+1$ UEs in total, denoted by $u_0,u_1,\ldots,u_{n+m}$,
respectively. For each $i\in[0,n+m]$, the home cell of UE $u_i$ is
$c_i$. For any $k\in[n+1,n+m]$, $u_k$ corresponds to a clause. The
numbers on the lines with arrows are the power gains of UEs.  The
gains not shown are negligible and treated as zero.}
\label{fig:np-hard}
\end{figure*}
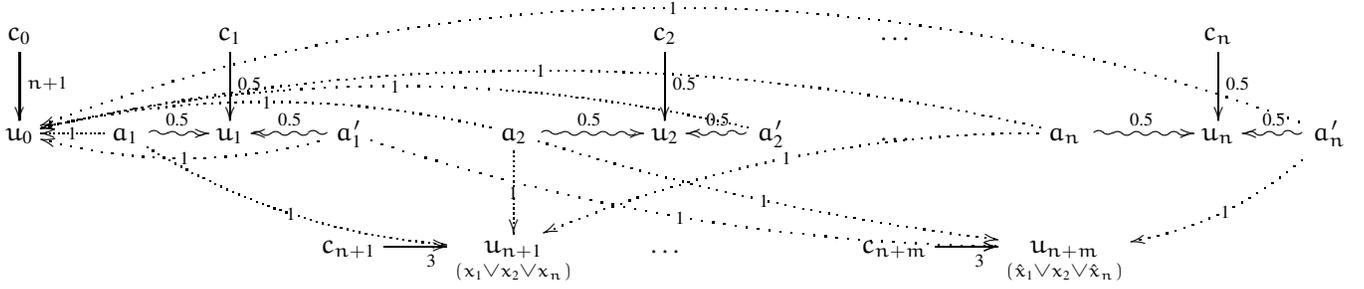

\begin{theorem}
Both \MSL{} and \MML{} are $\mathcal{NP}$-hard.
\label{thm:np-hard}
\end{theorem}

\begin{proof}
The idea is to use a polynomial-time reduction from the 3-satisfiability
($3$-\textit{SAT}) problem that is $\mathcal{NP}$-complete. Consider a
$3$-\textit{SAT} instance with $n$ Boolean variables $b_1,b_2,\ldots,b_n$, and
$m$ clauses. A variable or its negation is referred to as a literal. Denote by
$\hat{b}_i$ the negation of $b_i, i\in[1,n]$. Each clause is composed by a
disjunction of exactly three distinct literals, e.g., $(b_1\vee b_2\vee b_n)$.
The $3$-\textit{SAT} problem amounts to determining whether or not there exists
an assignment of true/false values to the variables, such that all clauses are
satisfied (i.e., at least one literal has value true in every clause).  For any
instance of $3$-\textit{SAT}, we construct a corresponding network scenario
shown in \figurename~\ref{fig:np-hard}. The transmit power of all cells, the
demand of all UEs, and the noise effect $\sigma^2$ are all uniformly set to
$1.0$. For all variables $b_1,b_2,\ldots,b_n$, we define $n$ UEs
$u_1,u_2,\ldots,u_n$. For the $m$ clauses, we define UEs
$u_{n+1},u_{n+2},\ldots,u_{n+m}$. For any $u_i,i\in[1,n+m]$, there is a home
cell $c_i$. According to the system model in Section~\ref{sec:system_model}, 
$u_i$ is associated with $c_i$, and possibly other cells.
For any variable $b_i$, we define two cells $a_i$ and $a'_i$,
representing the two literals and having symmetric gain for UE $u_i$.  There are
thus $3n+m$ cells in total.  The gain values between the cells and UEs are shown
in \figurename~\ref{fig:np-hard}. For any UE $u_i$, $i\in[1,n+m]$, the gain
values of $c_i$, $a_i$ and $a'_i$ equal $0.5$. For any $i\in[1,m]$, UE
$u_{n+i}$ has gain 1.0 from the cells that represent the literals in clause $i$,
whereas the gains from the other cells are negligible.  For example, the last
clause in \figurename~\ref{fig:np-hard}, i.e., clause $m$, is $(\hat{b}_1\vee
b_2\vee\hat{b}_n)$, defined for UE $u_{n+m}$. Then $u_{n+m}$ has non-negligible
gain from $a'_1$, $a_2$ and $a'_n$, and these gains are set to $1.0$. In
addition, we define an extra cell-UE pair $c_0$ and $u_0$. The gain values from
$a_i,a'_i, i \in{1,n}$ to $u_0$ are $1.0$. All other gains are negligible,
treated as zero and not shown in \figurename~\ref{fig:np-hard}.

We make several observations. For any $i\in[1,n]$, if $u_i$ is
only served by the home cell $c_i$, then $\gamma_{u_i}=0.5/1=0.5$ implying that
$x_{c_i}=1/\log_2(1+0.5)>1$ and hence the demand of $u_i$ cannot be
satisfied. For any $i\in[1,n]$, if $u_i$ is served by both $a_i$ and $a'_i$
besides the home cell $c_i$, then $c_0$ will be overloaded. To arrive at the
conclusion, observe that $\gamma_{u_i}=(0.5+0.5+0.5)/1=1.5$ resulting in
$x_{a_i}=x_{a'_i}=\log_2 2.5$. Then $\gamma_{u_0}=(n+1)/(2\log_2
2.5+(n-1)+1)<1$, leading to $x_{c_0}>1$. Therefore the demand of $u_0$ cannot be
met. Hence, for any pair $a_i$ and $a'_i$, one and exactly one of them will have
UE associated, and its load is $1/\log_2(1+(0.5+0.5)/1)=1.0$. For each clause, the
three cells corresponding to the literals of the clause cannot be all active in
serving UEs. Consider for example clause $(b_1 \vee b_2 \vee b_n)$.  If $a_1$,
$a_2$, and $a_n$ are all serving UEs, then
$\gamma_{u_{n+1}}=3/(1\times1.0+1\times1.0+1\times1.0+1)<1$, causing
$x_{c_{n+1}}>1$. However, we can verify that the demand of $u_{n+1}$ can be
satisfied if at most two of $a_1$, $a_2$ and $a_n$ have UEs to serve.

Suppose there is an association that ensures all the user's demands satisfied.
For each variable $b_i$, we set its value to be true if $a'_i$ is serving any
UE\@. Otherwise $a_i$ must be serving UEs instead, and we set $b_i$ to be false.
Now we evaluate the satisfiability of each clause. For convenience, consider
clause $(\hat{x}_1\vee x_2\vee\hat{x}_n)$ in \figurename~\ref{fig:np-hard}.
Since not all three cells $a'_1$, $a_2$ and $a'_n$ can be in the status of
serving UEs, at least one of $\hat{b}_1$, $b_2$ and $\hat{b}_n$ is true, and the
clause is true. Thus the
$3$-\textit{SAT} instance is feasible. Conversely, suppose we have a feasible solution for
the $3$-\textit{SAT} instance. Then we choose $a_i$ to serve, together with
$c_i$, UE $u_i$, if $\hat{b}_i$ is true.  If $\hat{b}_i$ is false, $a'_i$ is
chosen instead. Doing so satisfies the demands $u_0, u_1, \dots, u_n$.
Moreover, the demands $u_{n+1},u_{n+2},\ldots,u_{n+m}$ become satisfied as
well, because at most two out of the three cells defined for the three literals
of the clause will be serving any UE\@. Thus the association is feasible.
Hence the conclusion.  
\end{proof} 

Theorem~\ref{thm:np-hard} implies that no low-complexity and exact algorithm can
be expected, unless
$\mathcal{P}=\mathcal{NP}$. In addition, even for a sub-optimal algorithm, the
evaluation of candidate cell-UE association solutions is not straightforward,
due to the load-coupling constraints (\ref{eq:p0}b) and (\ref{eq:p1}b).  For
each candidate solution, obtaining the objective value requires to solve the
load-coupling equations, and this is typically done by fixed-point iterations.


\section{Problem Solving via Linear Approximation} 
\label{sec:linear}

\begin{figure}[t]
\centering
\includegraphics[width=\linewidth]{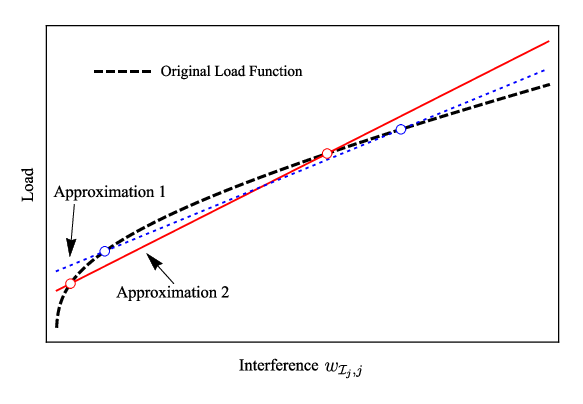}
\vspace{-4mm}
\caption{
Linear approximations for the concave load-coupling function. The
horizontal axis is the received interference and the vertical axis is the
load for serving a UE\@. The circles on each line denote the points used for construct the 
linear approximation.}
\label{fig:linear_approximation}
\end{figure}

In this section, we provide a solution approach 
based on linearization of the load function.
We use $f_{\mathcal{I}_j,j}$ to refer to the load of serving UE $j$ by
JT of cells in $\I_j$.  Note that this load occurs in all cells in
$\I_j$. The value of $f_{\mathcal{I}_j,j}$ is determined by the SINR,
which, in turn, depends on the amount of interference.  For UE $j$
served by cells in $\I_j$, we use $w_{\mathcal{I}_j,j}$ to denote the
interference variable, i.e., $w_{\mathcal{I}_j,j} =
\sum_{k\in\mathcal{I}\backslash\mathcal{I}_j}p_{k}g_{kj}x_{k}$.
It is easily proved that function $f_{\mathcal{I}_j,j}$ is concave
in $w_{\mathcal{I}_j,j}$ (cf.\ Lemma~\ref{lma:sif}), see
\figurename~\ref{fig:linear_approximation} for an illustration.  
The load-coupling equation is then as
follows.

\begin{equation}
\left\{
\begin{array}{ll}
x_i=\sum_{j\in\mathcal{J}_i}f_{\I_j,j}(w_{\I_j,j}) & \quad  i\in\mathcal{I}\\
w_{\I_j,j}=\sum_{k\in\mathcal{I}\backslash\mathcal{I}_j}p_{k}g_{kj}x_{k} & \quad
j\in\mathcal{J}
\end{array}
\right.
\label{eq:load_coupling2}
\end{equation}

Consider replacing $f(w_{\I_j,j})$ in
Eq.~(\ref{eq:load_coupling2}) by a linear function as an
approximation. The linear approximation can be obtained by taking two
points within the interference range of interest and evaluating the
nonlinear load function values, followed by constructing the
corresponding linear function that coincides with the nonlinear
function at these two points. Denote the two points by
$\undercirc{W}_{\I_j,j}$ and $\overcirc{W}_{\I_j,j}$, with
$\undercirc{W}_{\I_j,j}<\overcirc{W}_{\I_j,j}$, for UE $j$ and its
serving cell set $\I_j$. The approximation is illustrated in
\figurename~\ref{fig:linear_approximation}, for two different choices
of $\undercirc{W}_{\I_j,j}$ and $\overcirc{W}_{\I_j,j}$.  In the
figure, the $x$-axis is the interference variable $w_{\I_j,j}$, of
which the magnitude is typically $10^{-10}$ for a UE\@. We do not
explicitly show the values of the axis to keep the generality of the
illustration, because the interference depends highly on the specific
scenario and UE under consideration.

In equation form, the linear approximation reads

\vspace{-4mm}
\begin{equation}
l_{\mathcal{I}_j,j}(w_{\mathcal{I}_j,j})\triangleq s_{\mathcal{I}_j,j}w_{\mathcal{I}_j,j}+\mu_
{\mathcal{I}_j,j},
\label{eq:lower_bound} 
\end{equation}
where the two constants 
$s_{\I_j,j}$ and $\mu_{\I_j,j}$ are computed as follows.

\begin{equation}
s_{\mathcal{I}_j,j}\triangleq\frac{f_{\mathcal{I}_j,j}(\overcirc{W}_{\mathcal{I}_j,j})-f_{\mathcal
    {I}_j,j}(\undercirc{W}_{\mathcal{I}_j,j})}{\overcirc{W}_{\mathcal{I}_j,j}
-\undercirc{W}_{\mathcal{I}_j,j}}
\label{eq:slope}
\end{equation}

\begin{equation}
\mu_{\mathcal{I}_j,j}\triangleq
f_{\mathcal{I}_j,j}(\undercirc{W}_{\mathcal{I}_j,j})-
\undercirc{W}_{\mathcal{I}_j,j}s_{\I_j,j}
\label{eq:intercept}
\end{equation}

We use $\L_j$ to denote the set of possible association scenarios for
UE $j$; each element of $\L_j$ is a subset of the candidate set
$\I^+_j$.  Note that $|\L_j| = 2^{|\mathcal{I}_j^{+}|-1}$.  This is
because by minimum one cell has to be selected for association, and
the association must contain at least the home cell $c_j$ (which is an
element of $\I^+_j$). In the sequel, we use $\ell$ to index the
elements of $\L_j$. Note that because $\L_j$ is a set of sets, $\ell
\in \L_j$ is a set and corresponds to entity $\I_j$ in the above
introduction of linear approximation. The association variable
$\kappa_{ij}$ in \MSL{} is accordingly generalized to be
$\kappa_{\ell,j}$. We further introduce a constant $T_{\ell, j}$, defined as

\begin{equation}
T_{\ell, j} \triangleq \sum_{i \in \I \backslash \ell} p_{i}g_{ij}, ~~\ell \in \L_j, j \in \J  
\label{eq:big_T}
\end{equation}

Note that $T_{\ell, j}$, by definition, is an upper bound of interference, i.e.,
$w_{\ell,j} \leq T_{\ell, j}$, $\ell\in\mathcal{L}_j$,
$j\in\mathcal{J}$, because maximum load of all interfering cells is assumed in~\eqref{eq:big_T}.

Based on the linear approximation, we derive an MILP formulation of \MSL{}. The MILP formulation
is presented below.

\vspace{-2mm}
\begin{subequations} 
\begin{alignat}{2} 
[\textit{S--MILP}]~&\min\limits_{\bm{w},\bm{\kappa}} \quad
\sum_{i\in\mathcal{I}}x_i \\ 
\textnormal{s.t.}~& x_i=\sum\limits_{j \in \mathcal{J}}\sum\limits_{\ell\in\mathcal{L}_j:i\in\ell}\left(s_{\ell,j}w
_{\ell,j}+\mu_{\ell,j}\kappa_{\ell,j}\right) \\ \nonumber 
&
\quad\quad\quad\quad\quad\quad\quad\quad\quad\quad\quad\quad\quad\quad\quad
i\in\mathcal{I} \\ 
& 
w_{\ell,j}\geq\sum\limits_{i\in\mathcal{I}\backslash\ell}p_{i}g_{ij}x_i-T_{\ell, j}(1-\kappa_{\ell,j})
~ \\ \nonumber
&\quad\quad\quad\quad\quad\quad\quad\quad\quad\quad\quad\quad\!\!\!\ell\in\mathcal{L}_j,j\in\mathcal{J} \\ 
& w_{\ell,j}\geq{} 0
\quad\quad\quad\quad\quad\quad\quad\quad
{} \ell\in\mathcal{L}_j,j\in\mathcal{J}\\ 
& 0\leq{} x_i\leq{} 1
\quad\quad\quad\quad\quad\quad\quad\quad\quad\quad~
{} i\in\mathcal{I}\\
& \sum_{\ell\in{\mathcal{L}_j}}\kappa_{\ell,j}=1
\quad\quad\quad\quad\quad\quad\quad\quad\quad~
{} j\in\mathcal{J} \\
& \kappa_{\ell,j}\in\{0,1\}
\quad\quad\quad\quad\quad\quad~
{}\ell\in\mathcal{L}_j,j\in\mathcal{J} 
\end{alignat}
\label{eq:p2} 
\end{subequations}

Selecting one $\kappa_{\ell,j}~(\ell\in\mathcal{L}_j)$, as done in
Eq.~(\ref{eq:p2}f), corresponds to Eq.~(\ref{eq:p0}d) and Eq.~(\ref{eq:p0}e). 
The constraints shown in
Eq.~(\ref{eq:p2}b)--Eq.~(\ref{eq:p2}d) provide the effect of a linear
approximation of the non-linear constraint~(\ref{eq:p0}b). Note that
the desired term in the right-hand side of Eq.~(\ref{eq:p2}c) is
$\left( s_{\ell,j}w_{\ell,j}+\mu_{\ell,j}\right)\kappa_{\ell,j}$,
which is nonlinear because $w_{\ell,j}$ and $\kappa_{\ell,j}$ are
both variables. To see that we can achieve the same effect using
$\left(s_{\ell,j}w _{\ell,j}+\mu_{\ell,j}\kappa_{\ell,j}\right)$ and
the other constraints, consider separately the two cases
$\kappa_{\ell,j}=0$ and $\kappa_{\ell,j}=1$.  If $\kappa_{\ell,j}=0$,
Eq.~(\ref{eq:p2}c) becomes void, and hence there is no restriction on
$w_{\ell, j}$ except non-negativity.  Due to minimization, $w_{\ell,
j}=0$, and hence the entire term
$s_{\ell,j}w_{\ell,j}+\mu_{\ell,j}\kappa_{\ell,j}$ equals $0$, meaning
that $j$ does not impose any load of cells not serving the UE\@.  If
$\kappa_{\ell,j}=1$, then clearly $\left(s_{\ell,j}w
_{\ell,j}+\mu_{\ell,j}\kappa_{\ell,j}\right) = \left(s_{\ell,j}w
_{\ell,j}+\mu_{\ell,j}\right)\kappa_{\ell,j}$. In this case, UE $j$
is associated with cells in $\ell$, and $w_{\ell,j}\geq
\max\{\sum_{i\in\mathcal{I}\backslash\ell}p_{i}g_{ij}x_i, 0\}$ by
Eq.~(\ref{eq:p2}c) and Eq.~(\ref{eq:p2}d). Since the problem is
minimization, either Eq.~(\ref{eq:p2}c) or Eq.~(\ref{eq:p2}d) will
hold with equality. Thus we have
$w_{\ell,j}=\sum_{i\in\mathcal{I}\backslash\ell}p_{i}g_{ij}x_i$ that
indeed is the total amount of interference, provided
that the set of serving cells is $\ell$. Consequently the term
$s_{\ell,j}w_{\ell,j}+\mu_{\ell,j}\kappa_{\ell,j}$ represents the load
of the serving cells for UE $j$, computed by the linear function.
For \MML{}, the corresponding MILP formulation can be obtained via a slight
modification of Eq.~(\ref{eq:p2}a). We omit further details as the
modification is straightforward.
The solution approach of using the MILP formulation for \MSL{}
is summarized in Algorithm~\ref{alg:S--MILP}. 

\begin{algorithm}[tb]
\KwOut{$\bm{x}^{opt}$, $\bm{\kappa}^{opt}$}
\For{$j\in\mathcal{J}$, $\ell\in\mathcal{L}_j$}
{
Choose $\undercirc{W}_{\ell,j},\overcirc{W}_{\ell,j}\geq 0$, with
$\undercirc{W}_{\ell,j} < \overcirc{W}_{\ell,j}$\;
Compute $s_{\ell,j}$, $\mu_{\ell,j}$ by Eq.~(\ref{eq:slope}) and
Eq.~(\ref{eq:intercept})\;
}
Formulate \textit{S--MILP} as in Eq.~(\ref{eq:p2})\;
$\bm{\kappa}^{opt}\leftarrow$ optimum of \textit{S--MILP}\;
$\bm{x}^{opt}\leftarrow$ fixed point of
$\bm{f}(\bm{h}(\bm{x},\bm{\kappa}^{opt}),\bm{\kappa}^{opt})$\; 
\label{line:trueload}
\caption{MILP-based Load Minimization for \MML{}.}
\label{alg:S--MILP}
\end{algorithm}

The main advantage of using the linear approximation is to enable a
linear form of the optimization problem as an MILP, for which there
are standard optimization tools for problem solving to approach global
optimum. This serves two purposes. First, it yields an approximative
solution to the original problem, enabling a vis-a-vis comparison to
any algorithm that focuses on low complexity. Such an algorithm is
presented later in Section~\ref{sec:minl}. Second, and more
importantly, we will demonstrate how problem solving using linear
approximation leads to a bounding scheme, in terms of providing a
lower bound on the global minimum that effectively gauges the amount
of (worst-case) optimality gap and thereby performance assessment of
any sub-optimal algorithm. To these ends, Algorithm~\ref{alg:S--MILP}
is not distributed by nature. Rather, it is intended for use within the
cloud radio access network (C-RAN) architecture for future networks.

We remark that for solving \textit{S--MILP}, the standard approach
consists in branch-and-bound (and branch-and-cut), which scales
substantially better than an exhaustive search. Still, the scalability
is an issue, as the complexity, in the worst-case, is exponential; 
this justifies the derivation of the algorithm in
Section~\ref{sec:minl}. Note also that, for the purpose of bounding to
gain performance insights of other, low-complexity algorithms,
Algorithm~\ref{alg:S--MILP} is not intended to be run online.

Clearly, the solution from solving the linear optimization problem
varies by the chosen linear approximation which is defined by the
points $\undercirc{W}_{\I_j,j}$ and $\overcirc{W}_{\I_j,j}$.  In
general, these points should be close to the amount of interference
$w_{\I_j,j}$ at global optimum, which is of course not known a priori.
However, if one has a solution from any sub-optimal algorithm, e.g.,
the algorithm we derive in Section~\ref{sec:minl}, the output (i.e.,
expected interference) can be used to guide the selection. Moreover,
if $\undercirc{W}_{\I_j,j}$ and $\overcirc{W}_{\I_j,j}$ are chosen
such that the interference at optimum falls within the range
$[\undercirc{W}_{\I_j,j}, \overcirc{W}_{\I_j,j}]$, then solving the
problem using the approximation guarantees a lower bound of the global
minimum.  The trivial choice to achieve the bounding effect is to set
$\undercirc{W}_{\I_j,j}$ and $\overcirc{W}_{\I_j,j}$ to zero and the
most extreme of amount interference, respectively.  The latter
corresponds to assuming all other cells are interfering with full
load. However, as can be realized from
\figurename~\ref{fig:linear_approximation}, the gap between the 
linear function and the nonlinear one can be large for this choice of
range. Later in Section~\ref{sec:bounds}, we derive non-trivial yet
tractable solutions to reduce the range (while still guaranteeing its
validity) and thereby significantly strengthen the bound.

Instead of using the type of linear approximation that has been
discussed above, one can consider the first-order Taylor series
approximation, by constructing a linear function via the first-order
derivative of the nonlinear function at a selected point, such as the
middle point of $[\undercirc{W}_{\I_j,j}, \overcirc{W}_{\I_j,j}]$.
This linear approximation, although can be used in the MILP
formulation for a problem solution, does not enable a bounding scheme,
because the approximation over-estimates the true load value.  We also
remark that higher orders of Taylor series approximation result in
nonlinear objective functions that, from an optimization viewpoint, do
not exhibit advantage in comparison to the original nonlinear function.

\section{Load Minimization via Optimality Conditions} 
\label{sec:minl}

The second solution approach that we propose is based on improving
cell load by adjusting cell-UE association. This approach can be used
in conjunction with the MILP-based solutions, to further improve the
association obtained by solving the MILPs.  For a UE and its serving
cells, an adjustment amounts to either expanding the set of serving
cells with one additional cell, or removing one cell from the set of
serving cells. In the sequel, we use the term
\textit{link adjustment} to refer to the action of updating the cell-UE
association as described above.

\subsection{Link Adjustment: Basic Properties}

For link adjustment, we use $v$ and $u$ to respectively denote the
cell and UE under consideration.  The corresponding load and the SINR
functions are $f_v(\bm{x})$ and $h_u(\bm{\gamma})$, as 
defined in Eq.~\eqref{eq:sinr} and Eq.~\eqref{eq:load}, with the only
notational difference that $v$ and $u$ replace $i$ and $j$,
respectively, and hence the set of serving cells in Eq.~\eqref{eq:sinr} is
$\I_u$, and the set of served UEs in Eq.~\eqref{eq:load} is
$\J_v$.

Suppose $v \notin \I_u$, and the link adjustment operation is to
expand the set of serving cells to $\mathcal{I}_u\cup\{v\}$.
Consequently the set of cells of UE $u$ is reduced from
$\mathcal{I}\backslash\mathcal{I}_u$ to
$\mathcal{I}\backslash(\mathcal{I}_u\cup\{v\})$. We denote the vector
of SINR functions, resulted from the link adjustment, by $\bm{h}^+$,
with $\bm{h}^+=[h_1^{+}(\bm{x}),h_2^+(\bm{x}),\ldots,h_m^+(\bm{x})]$,
where for UE $u$, $h_u^{+}$ is given in
Eq.~\eqref{eq:adding_vartheta}, and $h_j^+(\bm{x})=h_j(\bm{x})$, for
all $j\neq u$, because the cell association remains for UEs other than
$u$.

\begin{equation} 
{h}^{+}_u(\bm{x})\triangleq\frac{\sum_{i\in
\mathcal{I}_u\bigcup\{v\}}p_{i}g_{iu}}{\sum_{k\in
\mathcal{I}\backslash(\mathcal{I}_j\cup\{v\})}p_{k}g_{ku}x_{ku}+\sigma^2}
\label{eq:adding_vartheta}
\end{equation}

For cell $v$, the set of associated UEs is expanded from
$\mathcal{J}_v$ to $\mathcal{J}_v\cup\{u\}$.  For all cells other than
$v$, the UE association remains.  Thus the vector of load function,
after the new, augmented association of UE $u$ to cell $v$, is
$\bm{f}^+(\bm{\gamma})=[f_1^1(\bm{\gamma}),f_2^+(\bm{\gamma}),\ldots,f_n^+(\bm
{\gamma})]$, where for cell $v$ the load function $f_v^{+}$ is
formulated in Eq.~(\ref{eq:adding_varphi}), and
$f_i^+(\bm{x})=f_i(\bm{x})$, for all $i\neq v$.

\vspace{-2mm}
\begin{equation}
{f}^{+}_v(\bm{\gamma})\triangleq\sum_{j\in\mathcal{J}_v\bigcup\{u\}}\frac{r_j}{MB\log
_2\left(1+\gamma_u\right)}
\label{eq:adding_varphi} 
\end{equation}

Next, we define the corresponding entities for 
removing one existing cell-UE association of UE $u$ and
cell $v$, assuming that $v$ is not the home cell of UE $u$.  That is,
the set of UE $u$'s serving cells is reduced from $\mathcal{I}_u$ to
$\mathcal{I}_u\backslash\{v\}$. The set of cells generating
interference to $u$ is expanded from
$\mathcal{I}\backslash\mathcal{I}_u$ to
$(\mathcal{I}\backslash\mathcal{I}_u)\cup\{v\}$, and the set of UEs
associated with cell $v$ is reduced from $\mathcal{J}_v$ to
$\mathcal{J}_v\backslash\{u\}$. The resulting SINR and
load functions are shown in Eq.~(\ref{eq:removing_vartheta}) and
Eq.~(\ref{eq:removing_varphi}), and denoted by $h^{-}_u$ and
$f^{-}_v$, respectively. The two vectors of functions are consistently
denoted by $\bm{h^-}$ and $\bm{f^-$}. Note that for any cell $i \not=
v$, $f^{-}_{i} (\bm{\gamma})=f_{i} (\bm{\gamma})$, and for any UE $j
\not=u$, $h^{-}_{j} (\bm{x})=h_{j} (\bm{x})$.

\begin{equation} 
{h}^{-}_u(\bm{x})\triangleq\frac{\sum_{i\in
\mathcal{I}_u\backslash\{v\}}p_{i}g_{iu}}{\sum_{k\in(\mathcal{I
}\backslash\mathcal{I}_u)\cup\{v\}}p_{k}g_{ku}x_{ku}+\sigma^2}
\label{eq:removing_vartheta} \end{equation} \begin{equation}
{f}^{-}_v(\bm{\gamma})\triangleq
\sum_{j\in\mathcal{J}_v\backslash\{u\}}\frac{r_j}{MB\log_2\left(1+\gamma
_u\right)}
\label{eq:removing_varphi} 
\end{equation}

Given the above definitions, the load-coupling equations after
expanding respectively reducing one cell-UE association via a link
adjustment of cell $v$ and UE $u$, are provided in
Eq.~(\ref{eq:load_coupling_plus}) and
Eq.~(\ref{eq:load_coupling_minus}).  Here, notation $\circ$ refers to
the compound function, i.e., $\bm{f}\circ\bm{h}(\cdot)$ means
$\bm{f}(\bm{h}(\cdot))$.

\begin{equation} 
\bm{x}=\bm{{f}^{+}}\circ\bm{{h}^{+}}~(\bm{x})
\label{eq:load_coupling_plus} 
\end{equation}
\vspace{-3mm}
\begin{equation} 
\bm{x}=\bm{{f}^{-}}\circ\bm{{h}^{-}}~(\bm{x})
\label{eq:load_coupling_minus} 
\end{equation}

We are interested in whether or
not the fixed points of above improve the load given by the fixed point of $\bm{f} \circ \bm{h} (\bm{x})$. 
To this end, we derive and 
prove conditions for load improvement.

\begin{lemma} 
The following properties hold true.  
\begin{enumerate} 
\item $\bm{{f}} \circ\bm{{h}^{+}}(\bm{x})\leq\min\left\{\bm{{f}
}\circ\bm{{h}}(\bm{x}),\bm{{f}}^{+}\circ\bm{{h}}^{+}(\bm{x})\right\}, 
\bm{x} \in \mathbb{R}^n_+$.  \item
    $\bm{{f}}\circ\bm{{h}}^{-}(\bm{x})\geq\max\left\{\bm{{f}}\circ\bm{{h}}(\bm{x}),\bm{{f}}^{-}\circ\bm{{h}}^{-}(\bm{x})\right\},
     \bm{x} \in \mathbb{R}^n_+$.
\end{enumerate}
\label{lma:x_bounds}	
\end{lemma}
\begin{proof} 
For 1), by the monotonicity of $\bm{f}$ and that
$\bm{{h}}^{+}(\bm{x})\geq\bm{{h}}(\bm{x})$ with strict inequality for UE $u$,
we obtain $\bm{{f}}\circ\bm{{h}}^{ +}(\bm{x})\leq\bm{{f}}\circ\bm{{h}}(\bm{x})$.
Function $\bm{h}$ exhibits monotonicity as well, and
$\bm{{f}}^{+}(\bm{\gamma})\geq\bm{{f}}(\bm{\gamma})$ for any $\bm{\gamma} \in
\mathbb{R}^m_+$, with strict inequality for cell $v$, thus $\bm{{f}}^{+
}\circ\bm{{h}}^{+}(\bm{x})\geq\bm{f}\circ\bm{h}^+(\bm{x})$. Statement 2) can be
proved analogously.
\end{proof}

\begin{lemma} 
The following properties hold true.  
\begin{enumerate} 
\item
$\bm{{h}}\circ\bm{{f}}^{+}(\bm{\gamma})\leq\min\{\bm{{h}
}\circ\bm{{f}}(\bm{\gamma}),\bm{{h}}^{+}\circ\bm{{f}^{+}}(\bm{\gamma})\},
 \bm{\gamma} \in \mathbb{R}^m_+$.
\item
$\bm{{h}}\circ\bm{{f}}^{-}(\bm{\gamma})\geq\max\{\bm{{h}
}\circ\bm{{f}}(\bm{\gamma}),\bm{{h}}^{-}\circ\bm{{f}^{-}}(\bm{\gamma})\},
 \bm{\gamma} \in \mathbb{R}^m_+$.
\end{enumerate} 
\label{lma:gamma_bounds} 
\end{lemma} 
\begin{proof} 
Because $\bm{h}$ has the monotonicity property and 
$\bm{{f}}^{+}(\bm{\gamma})\geq\bm{{f}}(\bm{\gamma})$ with strict 
inequality for cell $v$, we obtain $\bm{{h}
}\circ\bm{{f}}^{+}(\bm{\gamma})\leq\bm{{h}}\circ\bm{{f}}(\bm{\gamma})$.
Utilizing the monotonicity of 
$\bm{f}$ and 
$\bm{{h}}^{+}(\bm{x})\geq\bm{{h}}(\bm{x})$ for any $x\in\mathbb{R}^n_+$ with
strict inequality for UE $u$, we have $\bm{{h}}^{+}\circ\bm{{f
}}^{+}(\bm{\gamma})\geq\bm{h}\circ\bm{f}^+(\bm{\gamma})$. The observations lead
to statement 1), and statement 2) can be proved analogously.
\end{proof}

\subsection{Conditions for Cell Load Reduction}
\label{sec:condition}

We prove several conditions under which the load improves.
In Section~\ref{sec:design}, these conditions will be embedded into
the link-adjustment optimization algorithm later on.

\begin{definition}
Given functions $\bm{f}$, $\bm{f}^{+}$, $\bm{f}^{-}$ and $\bm{h}$,
$\bm{h}^{+}$, $\bm{h}^{-}$, we define the following notation.
\begin{itemize}
\item Denote by $\tilde{\bm{x}}$ the fixed point of function $\bm{f}\circ\bm{h}$,
    i.e., $\tilde{\bm{x}}$ is solution of $\bm{x}=\bm{{f}}\circ\bm{{h}}({\bm{x}})$, and denote
    $\tilde{\bm{\gamma}} \triangleq \bm{h}(\tilde{\bm{x}})$.
\item Denote by $\overplus{\bm{x}}$ the fixed point of function
$\bm{f}^{+}\circ\bm{h}^{+}$, i.e., $\overplus{\bm{x}}$ is the solution of
${\bm{x}}=\bm{f}^{+}\circ\bm{h}^{+}({\bm{x}})$, and 
denote $\overplus{\bm{\gamma}} \triangleq \bm{h}^{+}(\overplus{\bm{x}})$.
\item Denote by $\overminus{\bm{x}}$ the fixed point of function
$\bm{f}^{-}\circ\bm{h}^{-}$, i.e., $\overminus{\bm{x}}$ is the solution of
${\bm{x}}=\bm{f}^{-}\circ\bm{h}^{-}({\bm{x}})$, and 
denote $\overminus{\bm{\gamma}} \triangleq \bm{h}^{-}(\overminus{\bm{x}})$.
\end{itemize}
\label{def:notations_for_proofs}
\end{definition}

\begin{theorem} 
Consider the sequence $\bm{x}^{(0)}$, $\bm{x}^{(1)}$, $\dots$, where
$\bm{x}^{(0)}=\tilde{\bm{x}}$, and
$\bm{x}^{(t)}=\bm{{f}}\circ\bm{{h}}^{+}(\bm{x}^{(t-1)}), t\geq 1$.
 Then $\overplus{\bm{x}}\leq\tilde{\bm{x}}$, 
if for any iteration $k \geq 1$ we have $f_v^{+}\circ\bm{h}^{+}(\bm{x}^{(k)})\leq
x^{(k)}_v$.
\label{thm:adding_sufficient} 
\end{theorem}

\begin{proof}
Suppose $f_v^{+}\circ\bm{h}^{+}(\bm{x}^{(k)})\leq
x^{(k)}_v$ for iteration $k$.
By Lemma~\ref{lma:x_bounds}, we have $
\bm{x}^{(1)}=\bm{{f}}\circ\bm{{h}}^{+}(\bm{x}^{(0)})\leq\bm{{f}}\circ\bm{{h}}
(\bm{x}^{(0)})$. Since $\bm{x}^{(0)}=\tilde{\bm{x}}$,
$\bm{x}^{(0)}=\bm{f}\circ\bm{h}(\bm{x}^{(0)})$ holds. Thus
we obtain $\bm{x}^{(1)}\leq\bm{x}^{(0)}$. 
Therefore, by Property~1, 
\begin{equation} 
\bm{x}^{(k+1)}\leq\bm{x}^{(k)}\leq\bm{x}^{(k-1)}\leq\cdots\leq\bm{x}^{(1)}\leq\bm{x}^{(0)}
\label{eq:k_leq_k-1}
\end{equation}
To obtain $\overplus{\bm{x}}$, consider solving equation
$\bm{x}=\bm{f}^{+}\circ\bm{h}^{+}(\bm{x})$ by fixed-point
iterations. Let $\overplus{\bm{x}}^{(0)}$ be the initial point, with
$\overplus{\bm{x}}^{(0)}=\bm{x}^{(k)}$, and define the generic
iteration by $\overplus{\bm{x}}^{(t)}=\bm{f}^{+}\circ\bm{{h}}^{+}(\overplus{\bm{x}}^{(t-1)})$
for $t \geq 1$.
For cell $v$, we have
$\overplus{x}_{v}^{(1)}=f_v^{+}\circ\bm{h}^{+}(\overplus{\bm{x}}^{(0)})=
f_v^{+}\circ\bm{h}^{+}(\bm{x}^{(k)})$. 
Note that $f_v^{+}\circ\bm{h}^{+}(\bm{x}^{(k)})\leq x_v^{(k)}$.
By the construction $\overplus{\bm{x}}^{(0)}=\bm{x}^{(k)}$ we have
$f_v^{+}\circ\bm{h}^{+}(\bm{x}^{(k)})\leq x^{(k)}_v=\overplus{x}^{(0)}_v$. 
Therefore we obtain
\begin{equation}
\overplus{x}_v^{(1)}\leq\overplus{x}^{(0)}_v 
\label{eq:plus_v}
\end{equation}
For any cell $i\neq v$, recall that 
${f}^{+}_i\circ\bm{{h}}^{+}(\bm{x})={f}_i\circ\bm{{h}}^{+}(\bm{x})$ for any
$\bm{x}\in\mathbb{R}^n_{+}$. Thus
$\overplus{x}_{i}^{(1)}={f}^{+}_i\circ\bm{{h}}^{+}(\overplus{\bm{x}}^{(0)})={f}_i\circ\bm{{h}}^{+}(\overplus{\bm{x}}^{(0)})=f_{i}\circ\bm{h}^{+}(\bm{x}^{(k)})=x^{(k+1)}_i$.
Combined with Eq.~(\ref{eq:k_leq_k-1}), we have 
\begin{equation}
\overplus{x}_i^{(1)}\leq x_i^{(0)}\quad i\neq v
\label{eq:plus_i}
\end{equation}

Eq.~(\ref{eq:plus_v}) and Eq.~(\ref{eq:plus_i}) lead to the
conclusion that $\overplus{\bm{x}}^{(1)}\leq\overplus{\bm{x}}^{(0)}$. 
By Property 1, at convergence the following is true.
\begin{equation}
\overplus{\bm{x}}\leq\cdots\leq\overplus{\bm{x}}^{(2)}\leq\overplus{\bm{x}}^{(1)}
\leq\overplus{\bm{x}}^{(0)}=\bm{x}^{(k)}
\label{eq:x_converge}	
\end{equation}

Combined with Eq.~(\ref{eq:k_leq_k-1}), we obtain
$\overplus{\bm{x}}\leq\tilde{\bm{x}}$. Hence the conclusion.
\end{proof}

Theorem~\ref{thm:adding_sufficient} provides a partial optimality
condition for an optimization algorithm to examine whether or not
adding a link improves the performance. Let $\bm{x}^{(0)}$ be the load
vector before adding the link between cell $v$ and UE $u$, and
consider the sequence $\bm{x}^{(0)},\bm{x}^{(1)},\ldots$, with
$\bm{x}^{(t)}=\bm{f}\circ\bm{h}^{+}(\bm{x}^{(t-1)})$.  If the
condition in the theorem holds, the load of all cells could be reduced
by adding the link, yielding better objective value for both
\textit{MinSumL} and \textit{MinMaxL}.

\begin{theorem} 
Consider the sequence $\bm{\gamma}^{(0)}$, $\bm{\gamma}^{(1)}$, $\dots$, where
$\bm{\gamma}^{(0)} = \tilde{\bm{\gamma}}$, and   
$\bm{\gamma}^{(t)}=\bm{h}\circ\bm{f}^{-}(\bm{\gamma}^{(t-1)}), t\geq 1$.
Then $\overminus{\bm{x}}\leq\tilde{\bm{x}}$, if for any iteration $k \geq 1$ we have 
$h_u^{-}\circ\bm{f}^{-}(\bm{\gamma}^{(k)})\geq\gamma_u^{(k)}$.
\label{thm:removing_sufficient} 
\end{theorem}

\begin{proof} 
Suppose $h_u^{-}\circ\bm{f}^{-}(\bm{\gamma}^{(k)})\geq\gamma_u^{(k)}$ for iteration $k$. By
Lemma~\ref{lma:gamma_bounds}, we have
$\bm{\gamma}^{(1)}=\bm{h}\circ\bm{f}^{-}(\bm{\gamma}^{(0)})
\geq\bm{h}\circ\bm{f}(\bm{\gamma}^{(0)})$. Since 
$\bm{\gamma}^{(0)}=\tilde{\gamma}$,
$\bm{\gamma}^{(0)}=\bm{h}\circ\bm{f}(\bm{\gamma}^{(0)})$ holds. Hence
$\bm{\gamma}^{(1)}\geq\bm{\gamma}^{(0)}$. Therefore, by
Property~1, 
\begin{equation}
\bm{\gamma}^{(k+1)}\geq\bm{\gamma}^{(k)}\geq\bm{\gamma}^{(k-1)}\geq\cdots
\geq\bm{\gamma}^{(1)}\geq\bm{\gamma}^{(0)}.
\label{eq:k_geq_k-1}
\end{equation}

To obtain $\overminus{\bm{\gamma}}$, consider solving
$\bm{\gamma}=\bm{h}^{-}\circ\bm{f}^{-}(\bm{\gamma})$ by fixed-point
iterations. Let $\overminus{\bm{\gamma}}^{(0)}$ be the initial point,
with $\overminus{\bm{\gamma}}^{(0)}=\bm{\gamma}^{(k)}$, and define the
generic iteration by
$\overminus{\bm{\gamma}}^{(t)}=\bm{h}^{-}\circ\bm{f}^{-}(\overminus{\bm{\gamma}}^{(t-1)})$,
for $t \geq 1$.  For UE $u$, we have
$\overminus{\gamma}_u^{(1)}=h_u^{-}\circ\bm{f}^{-}(\overminus{\bm{\gamma}}^{(0)})
=h_u^{-}\circ\bm{f}^{-}(\bm{\gamma}^{(k)})$. Note that
$h_u^{-}\circ\bm{f}^{-}(\bm{\gamma}^{(k)})\leq\gamma_u^{(k+1)}$. Combined
with the construction $\overminus{\bm{\gamma}}^{(0)}=\bm{\gamma}^{(k)}$, we have
$h_u^{-}\circ\bm{f}^{-}(\bm{\gamma}^{(k)})\leq\gamma_{u}^{(k)}=
\overminus{\gamma}_{u}^{(0)}$. Therefore 
\begin{equation}
\overminus{\gamma}_u^{(1)}\geq\overminus{\gamma}_u^{(0)}
\label{eq:minus_u}
\end{equation}

For any UE $j\neq u$, recall that 
$h_{j}^{-}\circ\bm{f}^{-}(\bm{\gamma})=h_{j}\circ\bm{f}^{-}(\bm{\gamma})$ for
any $\bm{\gamma}\in\mathbb{R}^m_{+}$. Thus
$\overminus{\gamma}_j^{(1)}=h_j^{-}\circ\bm{f}^{-}(\overminus{\bm{\gamma}}^{(0)})=
h_j\circ\bm{f}^{-}(\overminus{\bm{\gamma}}^{(0)})=h_j\circ\bm{f}^{-}(\bm{\gamma}^{(k)})
=\gamma_j^{(k+1)}$. Combined with Eq.~(\ref{eq:k_geq_k-1}), we have
\begin{equation}
\overminus{\gamma}_j^{(1)}\geq\gamma_j^{(0)}\quad j\neq u
\label{eq:minus_j}
\end{equation}

By Eq.~(\ref{eq:minus_u}) and Eq.~(\ref{eq:minus_j}), one concludes
that $\overminus{\bm{\gamma}}^{(1)}\geq\overminus{\bm{\gamma}}^{(0)}$. By
Property 1, at convergence the following holds.
\begin{equation}
    \overminus{\bm{\gamma}}\geq\cdots\geq\overminus{\bm{\gamma}}^{(2)}\geq\overminus{\bm{\gamma}}^{(1)}
    \geq\overminus{\bm{\gamma}}^{(0)}=\bm{\bm{\gamma}}^{(k)}
\label{eq:gamma_converge}
\end{equation}
Combined with Eq.~(\ref{eq:k_geq_k-1}), we obtain
\begin{equation}
\overminus{\bm{\gamma}}\geq\tilde{\bm{\gamma}}
\end{equation}
Utilizing the above inequality and Lemma~\ref{lma:gamma_x}, we have
\begin{equation}
\label{eq:gammageq}
\bm{f}(\overminus{\bm{\gamma}})\geq\bm{f}(\tilde{\bm{\gamma}})
\end{equation}
By Lemma~\ref{lma:vice_versa}, $\overminus{\bm{x}}=\bm{f}(\overminus{\bm{\gamma}})$ and 
$\tilde{\bm{x}}=\bm{f}(\tilde{\bm{\gamma}})$. These, together 
with Eq.~\eqref{eq:gammageq}, give the conclusion.
\end{proof}

Similar to Theorem~\ref{thm:adding_sufficient},
Theorem~\ref{thm:removing_sufficient} can be used algorithmically, to
examine if an improvement can be obtained from removing a link between
a cell-UE pair.

The following corollaries provide two necessary conditions of load reduction.

\begin{corollary} 
For sequence $\bm{\gamma}^{(0)}$, $\bm{\gamma}^{(1)}$, $\dots$,
where $\bm{\gamma}^{(0)} =\tilde{\bm{\gamma}}$, and
$\bm{\gamma}^{(t)}=\bm{h}\circ\bm{f}^{+}(\bm{\gamma}^{(t-1)}), t\geq 1$,
$\overplus{\bm{x}}\leq\tilde{\bm{x}}$ with
strict inequality for at least one element, 
only if $\forall t\geq 1$,
$h_u^{+}\circ\bm{f}^{+}(\bm{\gamma}^{(t)})>\gamma_u^{(t)}$.
\label{thm:adding_necessary} 
\end{corollary}
\begin{proof} 
Suppose there is $k\geq 1$ such
that $h_u^{+}\circ\bm{f}^{+}(\bm{\gamma}^{(k)})\leq\gamma_u^{(k)}$. 
Following the flow of arguments in the proof of Theorem~\ref{thm:removing_sufficient},
$\overplus{\bm{x}}\geq\tilde{\bm{x}}$. Hence the conclusion.
\end{proof}

\begin{corollary} 
For sequence $\bm{x}^{(0)}$, $\bm{x}^{(1)}$, $\dots$, where
$\bm{x}^{(0)} = \tilde{\bm{x}}$, and $\bm{x}^{(t)}
=\bm{f}\circ\bm{h}^{-}(\bm{x}^{(t-1)}), t\geq 1$,
$\overminus{\bm{x}}\leq\tilde{\bm{x}}$ 
with strict inequality for at least one element,
only if $\forall t\geq 1$, $f_v^{-}\circ\bm{h}^{-}(\bm{x}^{(t)})<x_v^{(t)}$.
\label{thm:removing_necessary} 
\end{corollary}
\begin{proof} 
The proof is easily obtained
by the proof of Theorem~\ref{thm:adding_sufficient} and contradiction.
\end{proof}

From an algorithmic standpoint, Corollary~\ref{thm:adding_necessary}
and Corollary~\ref{thm:removing_necessary} complement
Theorem~\ref{thm:adding_sufficient} and
Theorem~\ref{thm:removing_sufficient}. Namely, as the two corollaries
provide necessary conditions, they can be used in an algorithm in
order to conclude link adjustments that will not improve the performance.

The derived conditions do not necessarily require to
account for all cells and UEs. This fact potentially enables to
examine the conditions to a cell and its local environment. Suppose
we add a link for cell $v$ and UE $u$. For any subset $\breve{\mathcal{I}}
\subset \I$ with $v\in\breve{\I}$, and the relevant set of UEs
$\breve{\mathcal{J}}=\bigcup_{i\in\breve{\mathcal{I}}}\mathcal{J}_i$, consider the
sufficient condition in Theorem~\ref{thm:adding_sufficient}. We argue
that the condition can be used, even if we fix the load of all cells
in $\mathcal{I}\backslash\breve{\mathcal{I}}$ and perform the so called
asynchronous fixed-point iteration~\cite{Yates:1995eh} to function
$\bm{f}\circ\bm{h}^{+}(\bm{x})$ for the load of cells in
$\breve{\mathcal{I}}$, as formalized below.

\begin{theorem}
Consider the sequence $\breve{\bm{x}}^{(0)}$, $\breve{\bm{x}}^{(1)}$, $\dots$, 
generated by the following asynchronous fixed-point iterations. 
\begin{enumerate}
\item $\breve{\bm{x}}^{(0)}=\tilde{\bm{x}}$.
\item For any $i\in\mathcal{I}\backslash\breve{\mathcal{I}}$ and $t\geq 1$, 
$\breve{x}_i^{(t)}=\breve{x}_i^{(0)}$.
\item For any $i\in\breve{\mathcal{I}}$ and $t\geq 1$, $\breve{x}_i^{(t)}=
f_i\circ\bm{h}^{+}(\breve{\bm{x}}^{(t-1)})$.
\end{enumerate}
Then $\overplus{\bm{x}}\leq\tilde{\bm{x}}$ if for some iteration $k\geq 1$ we have 
$f_v^{+}\circ\bm{h}^{+}(\breve{\bm{x}}^{(k)}) < \breve{x}_v^{(k)}$. 
\label{thm:local}
\end{theorem}

\begin{proof}
Suppose there is an iterate $\breve{\bm{x}}^{(k)}$, for some $k \geq
1$ such that
$f_v^+\circ\bm{h}^{+}(\breve{\bm{x}}^{(k)})<\breve{x}_v^{(k)}$.  By
Lemma~\ref{lma:x_bounds}, we have
$f_i\circ\bm{h}^{+}(\breve{\bm{x}}^{(0)})\leq
f_i\circ\bm{h}(\breve{\bm{x}}^{(0)})$, which leads to
$\breve{x}_i^{(1)}\leq\breve{x}_i^{(0)}$, for any
$i\in\breve{\mathcal{I}}$.  Combined with condition 2) that
$\breve{x}_i^{(1)}=\breve{x}_i^{(0)}$ for any
$i\in\mathcal{I}\backslash\breve{\mathcal{I}}$, we conclude that
$\breve{\bm{x}}^{(1)}\leq\breve{\bm{x}}^{(0)}$. For any
$i\in\breve{\mathcal{I}}$, by applying Property~1 on
$\breve{x}_i^{(2)}$, we obtain
$\breve{x}_i^{(2)}\leq\breve{x}_i^{(1)}$.  Again, combine the result with
condition 2), and repeat the procedure for the remaining iterations, one can conclude
\begin{equation}
\breve{\bm{x}}^{(k+1)}\leq
\breve{\bm{x}}^{(k)}\leq\breve{\bm{x}}^{(k-1)}\leq\cdots\leq\breve{\bm{x}}^{(1)}
\leq\breve{\bm{x}}_i^{(0)}
\end{equation}

Suppose we apply regular (instead of asynchronous) fixed-point
iterations to $\bm{f}^+\circ\bm{h}^{+}(\bm{x})$, using
$\breve{\bm{x}}^{(k)}$ as the initial point. This gives a sequence
$\overplus{\bm{x}}^{(0)}$, $\overplus{\bm{x}}^{(1)}$, $\dots$, with
$\overplus{\bm{x}}^{(0)} = \breve{\bm{x}}^{(k)}$. For cell $v$,
$\overplus{x}_v^{(1)} < \overplus{x}_v^{(0)} = \breve{x}_v^{(k)}$,
because $f_v^+\circ\bm{h}^{+}(\breve{\bm{x}}^{(k)})<\breve{x}_v^{(k)}$
and $\overplus{x}_v^{(1)} =
f_v^+\circ\bm{h}^{+}(\breve{\bm{x}}^{(k)})$. For any cell $i \not=v$,
note that $f^+_i \circ \bm{h}^+ (\bm{x}) = f_i \circ
\bm{h}^+ (\bm{x}),  \bm{x} \in \mathbb{R}^n_+$.  Hence $f^+_i
\circ \bm{h}^+ (\overplus{\bm{x}}^{(0)}) = f_i \circ \bm{h}^+
(\overplus{\bm{x}}^{(0)})$.  Moreover, $f^+_i \circ \bm{h}^+
(\overplus{\bm{x}}^{(0)}) = \overplus{x}_i^{(1)}$, and $f_i \circ
\bm{h}^+ (\overplus{\bm{x}}^{(0)}) = \breve{x}_i^{(k)}$.  Thus
$\overplus{x}_i^{(1)} = \breve{x}_i^{(k)}$, and by construction
$\overplus{x}_i^{(0)} = \breve{x}_i^{(k)}$. In conclusion,
$\overplus{\bm{x}}^{(1)} \leq \overplus{\bm{x}}^{(0)} =
\breve{\bm{x}}^{(k)}$.  The conclusion and Property~1 lead to
$\overplus{\bm{x}} \leq
\tilde{\bm{x}}$. 
\end{proof}

\subsection{Algorithm Design}
\label{sec:design}

In Algorithm~\ref{alg:AL} and Algorithm~\ref{alg:RL}, we 
apply the conditions in Section~\ref{sec:condition}
to examine possible load reduction by one link adjustment.
Since the steps of two algorithms are similar to each other, 
we omit some details in the description of Algorithm~\ref{alg:RL}.

\begin{algorithm}[h!]
\KwIn{
$\tilde{\bm{x}}$,
$\tilde{\bm{\gamma}}$,
$\tilde{\bm{\kappa}}~(\tilde{\kappa}_{vu}=0)$, 
$\langle v,u\rangle$
}
\KwOut{$\bm{\kappa}$ or $\tilde{\bm{\kappa}}$}
$\bm{\kappa}\leftarrow\tilde{\bm{\kappa}}$; $\kappa_{vu}\leftarrow{1}$\;
$\bm{x}^{(0)}\leftarrow\tilde{\bm{x}}$; 
$\tilde{\bm{\gamma}}^{(0)}\leftarrow\tilde{\bm{\gamma}}$\; 
\For{$t\leftarrow 1$ \textnormal{\textbf{to}} $\tau$}
{
    $\bm{x}^{(t)}\leftarrow
    \bm{f}\big(\bm{h}(\bm{x}^{(t-1)},\bm{\kappa}),\tilde{\bm{\kappa}}\big)$\; \label{Line:x2}
    $\bm{\gamma}^{(t)}\leftarrow
    \bm{h}\big(\bm{f}(\bm{\gamma}^{(t-1)},\tilde{\bm{\kappa}}),\bm{\kappa}\big)$\; \label{Line:gamma2}
    \uIf{$f_v(\bm{h}(\bm{x}^{(t)},\bm{\kappa}),\bm{\kappa})\leq x_v^{(t)}$
\label{Line:alsuff}}  
    {
        \KwRet{$\bm{\kappa}$}\;
    }
    \uElseIf{${}h_u(\bm{f}(\bm{\gamma}^{(t)},\bm{\kappa}),\bm{\kappa})\leq\gamma_u^{(t)}$
\label{Line:alness}} 
    {
        \KwRet{$\tilde{\bm{\kappa}}$}\;

    }
}
\KwRet{$\tilde{\bm{\kappa}}$}\;
\caption{Link adjustment of adding one association.} 
\label{alg:AL}
\end{algorithm}

\begin{algorithm}[h!]
\KwIn{
$\tilde{\bm{x}}$,
$\tilde{\bm{\gamma}}$,
$\tilde{\bm{\kappa}}~(\tilde{\kappa}_{vu}=1)$, 
$\langle v,u\rangle$
}
\KwOut{$\bm{\kappa}$ or $\tilde{\bm{\kappa}}$}
The algorithm flow follows that of Algorithm~\ref{alg:AL}, with
$\kappa_{vu}\leftarrow 1$ changed to $\kappa_{vu}\leftarrow 0$ in Line
1, and adaptation of Lines 4, 5, 6 and 8 to the conditions
in Theorem~\ref{thm:removing_sufficient} and
Corollary~\ref{thm:removing_necessary}.
\caption{Link adjustment of removing one association.}
\label{alg:RL}
\end{algorithm}

The input of Algorithm~\ref{alg:AL} consists of an association
$\tilde{\bm{\kappa}}$, the load $\tilde{\bm{x}}$ and the SINR
$\tilde{\bm{\gamma}}$ for association $\tilde{\bm{\kappa}}$, and a
cell-UE pair $\langle v,u \rangle$, for which $\tilde{\kappa}_{vu}=0$ and
hence it is subject to the consideration of adding the association.
Parameter $\tau$ is pre-defined, and denotes the maximum number of
iterations of applying the conditions in
Theorem~\ref{thm:adding_sufficient} and
Corollary~\ref{thm:adding_necessary}.
Lines~\ref{Line:alsuff} and~\ref{Line:alness} apply, respectively, the
conditions provided in Theorem~\ref{thm:adding_sufficient} and
Corollary~\ref{thm:adding_necessary}.  If the outcome of
Line~\ref{Line:alsuff} turns out to be true, then adding the association
$\langle v,u \rangle$ improves the load of all cells by
Theorem~\ref{thm:adding_sufficient}, and Algorithm~\ref{alg:AL}
updates the cell-UE association to include this pair. If
Line~\ref{Line:alness} holds true, then the necessary condition for load
improvement in Corollary~\ref{thm:adding_necessary} is not met.
Consequently $\langle v,u\rangle$ is rejected and the original
association $\tilde{\bm{\kappa}}$ with $\tilde{\kappa}_{vu}=0$ is
kept. This association is also kept in case neither the sufficient nor
the necessary conditions is met. 
Note that in considering the removal of a link, the conditions
in Theorem~\ref{thm:removing_sufficient} 
and Corollary~\ref{thm:removing_necessary} are used 
in Algorithm~\ref{alg:RL}.

\begin{algorithm}[tb]
\KwIn{$\bm{\kappa}^{init}$}
\KwOut{$\bm{\kappa}^{opt}$}
$\tilde{\bm{\kappa}}\leftarrow\bm{\kappa}^{init}$\;
$\tilde{\bm{x}} \leftarrow \textnormal{Fixed point of }
\bm{f}(\bm{h}(\bm{x},\tilde{\bm{\kappa}}),\tilde{\bm{\kappa}})$\;
$\tilde{\bm{\gamma}} \leftarrow \bm{h}(\tilde{\bm{x}},\tilde{\bm{\kappa}})$\; 
\Repeat{$\lambda=0~\vee~\urcorner\mathsf{flag}$}{
    $\lambda\leftarrow\lambda-1$; $\mathsf{flag}\leftarrow$ false
    \;\label{Line:flag_init}
    \For{$ u\in\mathcal{I}$, $v\in\mathcal{J}$}
    {
        $\eta\leftarrow\tilde{\kappa}_{vu}$\;
        \eIf{$\eta=0$}{
            $\tilde{\bm{\kappa}}\leftarrow
            \textnormal{Algorithm~\ref{alg:AL} with~} (\tilde{\bm{x}},\tilde{\bm{\gamma}},\tilde{\bm{\kappa}},\langle v,u \rangle)$\;
        }{
             $\tilde{\bm{\kappa}}\leftarrow
             \textnormal{Algorithm~\ref{alg:RL} with~} (\tilde{\bm{x}},\tilde{\bm{\gamma}},\tilde{\bm{\kappa}},\langle v, u\rangle)$\;
        }
        \If{$\tilde{\kappa}_{vu}\neq\eta$
\label{Line:flag_update}}{
            $\mathsf{flag}\leftarrow$ true \;
            $\tilde{\bm{x}} \leftarrow \textnormal{Fixed point of } 
            \bm{f}(\bm{h}(\bm{x},\tilde{\bm{\kappa}}),\tilde{\bm{\kappa}})$\; \label{Line:fp}
            $\tilde{\bm{\gamma}} \leftarrow \bm{h}(\tilde{\bm{x}},\tilde{\bm{\kappa}})$\;  
        }
    }
}
$\bm{\kappa}^{opt}\leftarrow\tilde{\bm{\kappa}}$\;
$\bm{x}^{opt}\leftarrow\tilde{\bm{x}}$ 
\caption{Minimization of Load (\textit{MinL})} 
\label{alg:MinL}
\end{algorithm}

The overall algorithm, named \textit{MinL}, is presented in
Algorithm~\ref{alg:MinL}. The outer loop runs for $\lambda$ rounds, where
$\lambda$ is a pre-defined parameter.  The inner loop goes through all the
cell-UE pairs.  For each pair, Algorithm~\ref{alg:AL} or~\ref{alg:RL} is
utilized to evaluate the corresponding link adjustment.  The indicator 
$\mathsf{flag}$, is initialized to be false, at the beginning of each outer loop
round. The outcome of
Algorithm~\ref{alg:AL} and~\ref{alg:RL} is examined in
Line~\ref{Line:flag_update}. If there is an improvement, the load and SINR of
the updated association $\tilde{\bm{\kappa}}$ are computed, and $\mathsf{flag}$
is set to be true, before moving to the next cell-UE pair. For
each outer loop round, if there is no improvement throughout all its inner loop
rounds, then $\mathsf{flag}$ is not updated and kept false, which, causes the
outer loop to terminate.

As for algorithm complexity, we observe first
that the complexity of Algorithm~\ref{alg:AL} and
Algorithm~\ref{alg:RL}, in terms of evaluating fixed-point solutions,
is clearly of $O(\tau m +
\tau n)$. Recall that $m$ and $n$ are the numbers of users and cells,
respectively. In Algorithm~\ref{alg:MinL}, Algorithm~\ref{alg:AL} (or
Algorithm~\ref{alg:RL}) is invoked $\lambda mn$ times at maximum.
Also, the fixed point evaluations in Algorithm~\ref{alg:MinL} in
Line~\ref{Line:fp} are executed $\lambda mn$ times at most.  The
complexity of the loop in Algorithm~\ref{alg:MinL} is thus $O(\lambda
mn)\times O(\tau m + \tau n + Kn)=O(\lambda \tau m^2 n+\lambda\tau mn^2
+ \lambda mn^2K)$, where $K$ is the number of fixed-point iterations per
evaluation. For standard interference functions, the fixed-point
method has been shown to have geometric convergence and does not generate 
the issue of
computational bottleneck (see~\cite{Huang:1998uc}).  
Moreover, out of the loop of Algorithm~4, in Line 2, we obtain the fixed point 
of the load-coupling function under the initial association.
Therefore, the overall complexity of Algorithm~\ref{alg:MinL}
is $O(Kn+\lambda \tau m^2 n+\lambda\tau mn^2 + \lambda
mn^2K)$=$O(\lambda \tau m^2 n+\lambda\tau mn^2 + \lambda mn^2K)$, which is
polynomial in the input size, the number of fixed-point solutions, and
algorithm control parameters. Hence the algorithm is scalable.

We remark that \textit{MinL}, presented in Algorithm~\ref{alg:MinL},
suits well for distributed implementation. By Theorem~\ref{thm:local},
by means of asynchronous fixed-point iterations, the results of
Theorem~\ref{thm:adding_sufficient},
Theorem~\ref{thm:removing_sufficient},
Corollary~\ref{thm:adding_necessary}, and
Corollary~\ref{thm:removing_necessary}, apply also to any subset of
cells.  Therefore, for Algorithm~\ref{alg:MinL}, a distributed
implementation amounts to, for any specific association of cell $v$
and UE $u$, examining the partial optimality conditions for a cell subset
$\breve{\mathcal{I}}$ with $v \in
\breve{\mathcal{I}}$ (cf.~Theorem~\ref{thm:removing_sufficient}).
Such an implementation is hence straightforward in concept. Here, subset
$\breve{\mathcal{I}}$ is naturally formed by the neighboring cells of $v$,
because more remote cells have virtually no significance in
interference. Thus, the distributed implementation only requires
information exchange that is local to a cell and hence easily managed
via the X2 interface.

\section{Lower Bounds of Global Optimum}
\label{sec:bounds}

In this section, we show how to derive lower bounds of the global
optimum, for both \MSL{} and \MML{}, based on the linearization
in Section~\ref{sec:linear}. Such lower bounds provide effective means
to gauge of the performance of the solution approaches 
presented in Sections~\ref{sec:linear} and~\ref{sec:minl}.

Consider the linear approximation derived using the two end points $0$ and
$T_{\ell, j}$, where $j$ and $\ell\in\mathcal{L}_j$ are the UE and its set of
serving cells under consideration, and $T_{\ell, j}$ is defined in
Eq.~(\ref{eq:big_T}). The linear approximation is
illustrated in \figurename~\ref{fig:lower_bound}. Similar to
\figurename~\ref{fig:linear_approximation}, the specific values of the axis are not shown
in \figurename~\ref{fig:lower_bound} for the sake of generality. Constructing the linear approximation for
all UEs and their candidate serving cells, the resulting MILP
formulations \textit{S--MILP} and
\textit{M--MILP} will provide lower bounds of the global optimum of \MSL{} and
\MML{}, respectively. This is simply because, no matter the association, the
interference is always within the interval used for constructing the linear
function, and the load calculated by the linear function is an underestimation
of the load for the entire interval.

The gap between the original, nonlinear load function and the linear
function as constructed above can be very significant, as can be seen
in \figurename~\ref{fig:lower_bound}. To overcome the issue, we derive
a better approximation by reducing the domain of the interference
variable $w_{\ell,j}$, $\ell\in\mathcal{L}_j,j\in\mathcal{J}$, with the theoretical guarantee that
the interference, if UE $j$ is served by cells $\ell$ at the global
optimum, will indeed fall within the reduced domain. Consequently, the
linear function derived by the reduced interference domain remains an
underestimation of the load of the global optimum, whereas the quality
of the approximation is improved. To this end, we first derive lower
and upper bounds for the load vector $\bm{x}$ in
Lemma~\ref{lma:bound_load} below.

\begin{figure}[h!]
\centering
\includegraphics[width=\linewidth]{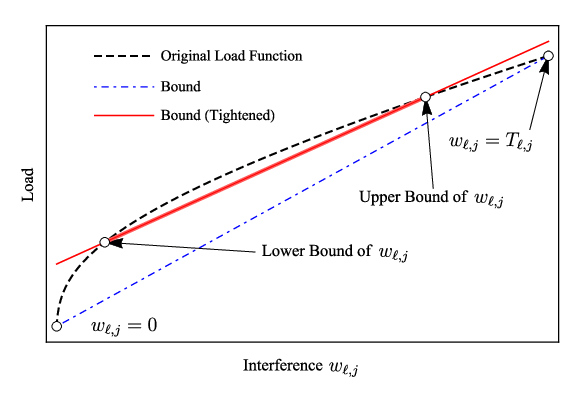}
\vspace{-5mm}
\caption{The lower bound can be improved by reducing the range of relevance 
for the interference variable.}
\label{fig:lower_bound}
\end{figure}
\begin{figure*} 
\centering 
\subfigure[Initial Association]{\includegraphics[width=0.4\linewidth]{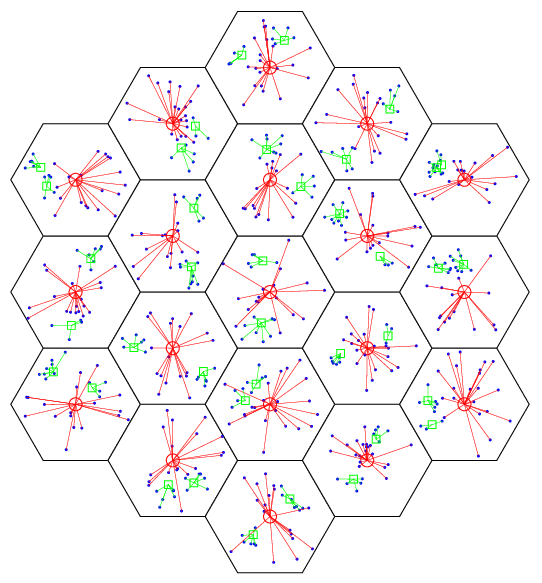}}
\quad\quad\quad 
\subfigure[Optimized Association]{\includegraphics[width=0.4\linewidth]{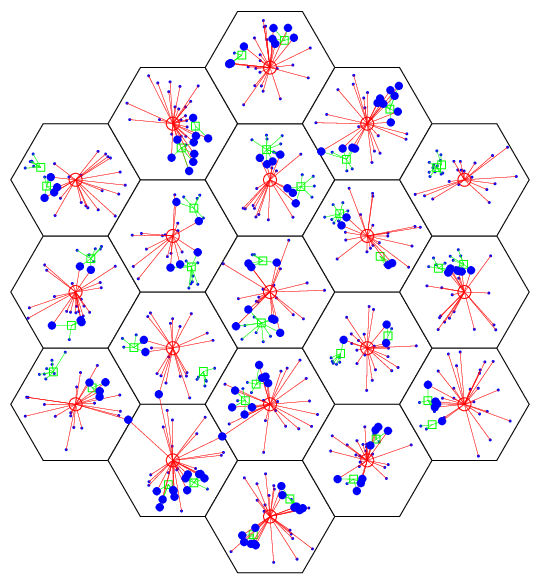}} 
\caption{An example of the HetNet layout, and initial and the optimized
association. Red circles in the center of each hexagon denote
the MC\@. Green rectangles denote the SC\@. The UEs are denoted by 
blue dots. \figurename~\ref{fig:layout}(a) shows the network layout and initial association 
, where each UE is associated only to its home cell. 
\figurename~\ref{fig:layout}(b) shows the optimized association, where UEs served by
multiple cells using JT are marked by thicker blue dots.} 
\label{fig:layout} 
\end{figure*}

\subsection{Bounds of Cell Load}

\begin{lemma}
Define two matrices $\hat{\bm{\kappa}}$ and $\check{\bm{\kappa}}$.
For $\hat{\bm{\kappa}}$, entry $\hat{\kappa}_{ij}=1$ if $i=c_j$, 
and zero otherwise. For $\check{\bm{\kappa}}$,
entry $\check{\kappa}_{ij}=1$ if $j\in\mathcal{J}^{+}_i$,
and zero otherwise.
Denote by $\undercirc{\bm{x}}$ and $\overcirc{\bm{x}}$ the fixed points
of ${\bm{x}}=\bm{f}(\bm{h}({\bm{x}},\check{\bm{\kappa}}),
\hat{\bm{\kappa}})$ and ${\bm{x}}=\bm{f}(\bm{h}({\bm{x}},\hat{\bm{\kappa}}),
\check{\bm{\kappa}})$, respectively.
For any cell-UE association, the incurred cell load vector is bounded
from below and above, respectively, by $\undercirc{\bm{x}}$ and
$\overcirc{\bm{x}}$. 

\label{lma:bound_load}
\end{lemma}

\begin{proof}
Consider first $\undercirc{\bm{x}}$.  For any given load vector
$\bm{x}$ and an association vector $\bm{\kappa}$, it can be verified
that
$\bm{h}(\bm{x},\check{\bm{\kappa}})\geq\bm{h}(\bm{x},\bm{\kappa})$.
This is because, for any UE $j\in\mathcal{J}$, all the candidate cells
in $\mathcal{I}^{+}_j$ are serving $j$ in $\check{\bm{\kappa}}$.
Moreover, for any $\bm{\gamma}$,
$\bm{f}(\bm{\gamma},\hat{\bm{\kappa}})\leq\bm{f}(\bm{\gamma},\bm{\kappa})$,
because for any cell $i\in\mathcal{I}$, the load $x_i$ originates only
from serving the UEs in $\mathcal{J}_i^-$, for which $i$ is the home
cell. Hence
$\bm{f}(\bm{h}(\bm{x},\check{\bm{\kappa}}),\hat{\bm{\kappa}})\leq\bm{f}\left(
\bm{h}\left(\bm{x},\bm{\kappa}\right),\bm{\kappa}\right)$ holds for any
$\bm{\kappa}$ and $\bm{x}$. This establishes the validity of
$\undercirc{\bm{x}}$ as lower bound of cell load.  The validity of
$\overcirc{\bm{x}}$ as upper bound can be proven analogously, by
observing that
$\bm{h}(\bm{x},\hat{\bm{\kappa}})\leq\bm{h}(\bm{x},\bm{\kappa})$ and
$\bm{f}(\bm{\gamma},\check{\bm{\kappa}})\geq\bm{f}(\bm{\gamma},{\bm{\kappa}})$,
which together lead to
$\bm{f}(\bm{h}(\bm{x},\hat{\bm{\kappa}}),\check{\bm{\kappa}})\geq\bm{f}(\bm{h}
(\bm{x},\bm{\kappa}),\bm{\kappa})$.
\end{proof}

Intuitively, the lower bound $\undercirc{\bm{x}}$ given by
Lemma~\ref{lma:bound_load} is rather loose, because it is derived by assuming
the most optimistic interference conditions which are fully unrealistic.
Nevertheless, the two bounds on $\bm{x}$ can be utilized together to reduce the
domain of interest for the interference variables, and thus improving the
approximation as discussed earlier.  One can refer to
\figurename~\ref{fig:lower_bound} for an illustration. Below, theoretical
insights are given, ensuring that the strengthened linear approximation remains
valid in delivering a lower bound to the global optimum.



\subsection{Validity of the Improved Lower Bound}


\begin{corollary} 
For any UE $j$ and $\ell\in\mathcal{L}_j$, 
$w_{\ell,j}\in[\undercirc{W}_{\ell,j},\overcirc{W}_{\ell,j
}]$, with
$\undercirc{W}_{\ell,j}=\sum_{i\in\mathcal{I}\backslash\ell}p_
{i}g_{ij}\undercirc{x}_i$ and
$\overcirc{W}_{\ell,j}=\sum_{i\in\mathcal{I}\backslash\ell}p_{i
}g_{ij}\overcirc{x}_i$.
\label{thm:bound_I} 
\end{corollary}

\begin{proof} 
Since both $p_i$ and $g_{ij}$ are constant for any $i\in\mathcal{I}$ and
$j\in\mathcal{J}$, this result follows directly from Lemma~\ref{lma:bound_load}. 
\end{proof}

\begin{theorem} 
(\textbf{Lower bound}) The optimum of \textit{S--MILP} is a lower bound of
\MSL{}, with the two constants $s_{\ell,j}$ and $\mu_{\ell,j}$ as defined in 
in Corollary~\ref{thm:bound_I}, and the following additional constraint
\begin{equation}
w_{\ell,j}\geq
\undercirc{W}_{\ell,j}-T(1-\kappa_{\ell,j})\quad
\ell\in\mathcal{L}_j,j\in\mathcal{J}
\label{eq:extra}
\end{equation}
\end{theorem}
\begin{proof} 
The extra constraint shown in Eq.~(\ref{eq:extra}) guarantees that we have
$w_{\ell,j}\geq\max\{\undercirc{W}_{\ell,j}$,
$\sum_{i\in\mathcal{I}\backslash\ell}p_{i}g_{ij}x_i\}$ for any
$\ell\in\mathcal{L}_j$ and $j\in\mathcal{J}$, if $\kappa_{\ell,j}=1$.  That
is, if $\sum_{i\in\I\backslash\ell}p_{i}g_{ij}x_{i}<\undercirc{W}_{\ell,j}$,
then $w_{\ell,j}$ is equal to its lower bound $\undercirc{W}_{\ell,j}$.
Otherwise, we have $w_{\ell,j}=\sum_{i\in\I\backslash\ell}p_{i}g_{ij}x_{i}$. 

Consider the optimal load solution in \MSL{}. For any UE $j$, suppose its
serving cells set is $\ell\in\mathcal{L}_j$.  By construction of the linear
function, the interference $w_{\ell,j}$ satisfies $l_{\ell,j}(w_{\ell,j})
\leq f_{\ell,j}(w_{\ell,j})$, for interval
$[{\undercirc{W}_{\ell,j}},\overcirc{W}_{\ell,j}]$.  Therefore, for the optimal
solution of \MSL{}, the objective value of \textit{S--MILP} is a lower bound with
the same association. This association solution is not necessarily the optimum
of \textit{S--MILP}.  Because \textit{S--MILP} is minimization, the theorem
follows.
\end{proof}

For \MML{} and its approximation \textit{M--MILP}, the validity of the lower
bound can be proved in the same way.

\section{Performance Evaluation} 
\label{sec:simulation}

\subsection{Simulation Settings}

The network scenario consists in 19 hexagonal regions (see
\figurename~\ref{fig:layout}), each of which has one MC in
the center. In each hexagon, two SCs and thirty UEs are
randomly distributed. The hexagon radius is 500~m. The HetNet
operates at 2~GHz. Each resource block follows the LTE standard of 180
kHz bandwidth and the total bandwidth is 20~MHz~\cite{3gpp}. The
transmit power per resource block for MCs and SCs equal 400~mW and
50~mW, respectively. The noise power spectral density is 
-174~dBm/Hz. The path loss for MCs and SCs follow the standard 3GPP
urban macro (UMa) model and urban micro (UMi) model of hexagonal
deployment, with the shadowing coefficients generated by the log-normal 
distribution with 6~dB and 3~dB standard deviation,
respectively~\cite{Anonymous:GtzCaaVU}. For any UE, we set 3 candidate
cells (MCs and/or SCs) and thus $|\mathcal{I}^{+}_j|=3,~j\in\mathcal{J}$.  
In the MILP-based solutions, the linearization
follows Corollary~\ref{thm:bound_I}.  In \textit{MinL}, the parameters
$\lambda$ and $\tau$ are set to 3 and 5, respectively.  As a baseline
solution, each UE is served only by the home cell (MC or 
SC) with the best received signal power, without any JT\@. 
\textit{MinL} is evaluated with two different initial
associations: 1) the baseline solution, and 2) the association
solution obtained from solving MILPs. The latter is labeled as
``$\textit{S--MILP}$ + \textit{MinL}'' and ``$\textit{M--MILP}$ +
\textit{MinL}''. The simulations run on 
20 data sets, and the 
results are averaged over the data sets.

\subsection{An Illustration of Optimized Association}

An example network layout is shown in \figurename~\ref{fig:layout}.
\figurename~\ref{fig:layout}(a) shows the initial association, and
\figurename~\ref{fig:layout}(b) shows the association optimized by
$\textit{S--MILP}$ + \textit{MinL}. The MCs are marked by red circles, 
the SCs are marked by green rectangles, and the UEs are represented by blue
dots. The red and green lines show the association of MCs and SCs to
UEs, respectively. In \figurename~\ref{fig:layout}(b), the UEs in JT
are marked by thicker dots. One can see from
\figurename~\ref{fig:layout} that most cooperations by JT are between MCs and
SCs, whereas JT of MCs occur for a few UEs on MC edges only.
From the figure, we infer that JT is likely to take place if 
at least two cells provide comparable received powers to a UE\@.
In addition, it can be seen that the number of UEs
in JT varies much over the area, reflecting the heterogeneity of the scenario. 
Also, although not shown in the figure, the simulations indicate that the
number of UEs in JT increases with respect to the demand.

\subsection{Performance Comparison for MinSumL} 
\label{sec:ressuml}

\begin{figure}[h!] 
\centering 
\subfigure[Sum Load (MC+SC)]{\includegraphics[width=\linewidth]{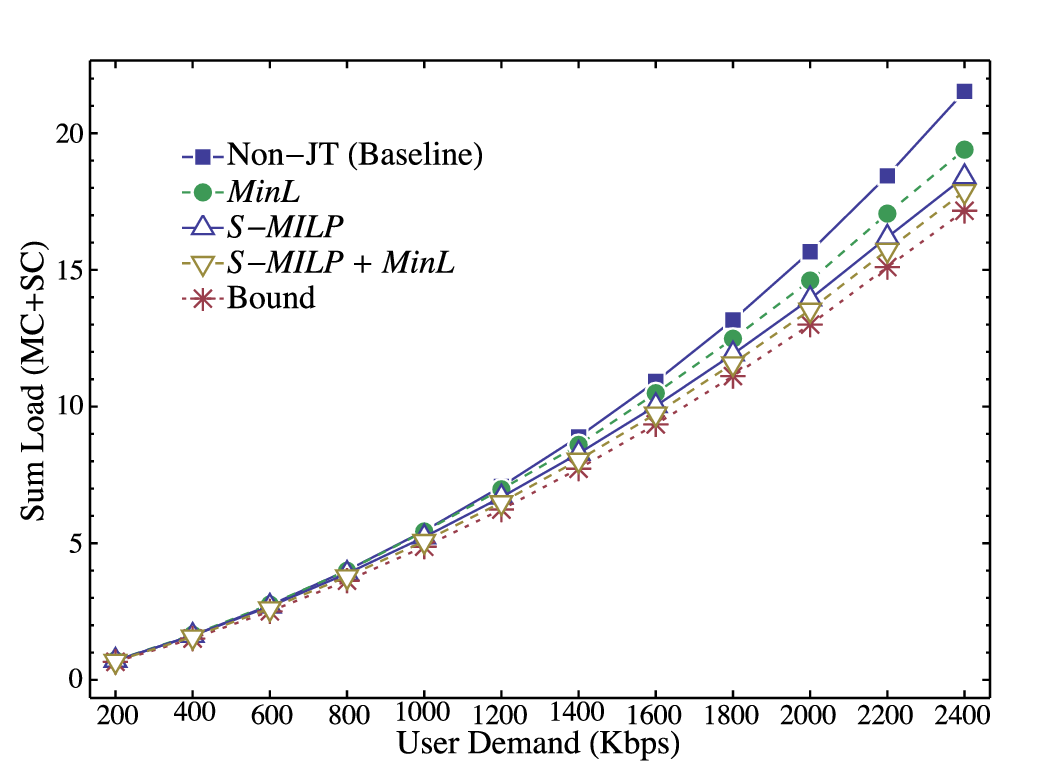}} \\
\subfigure[Sum Load (MC)]{\includegraphics[width=\linewidth]{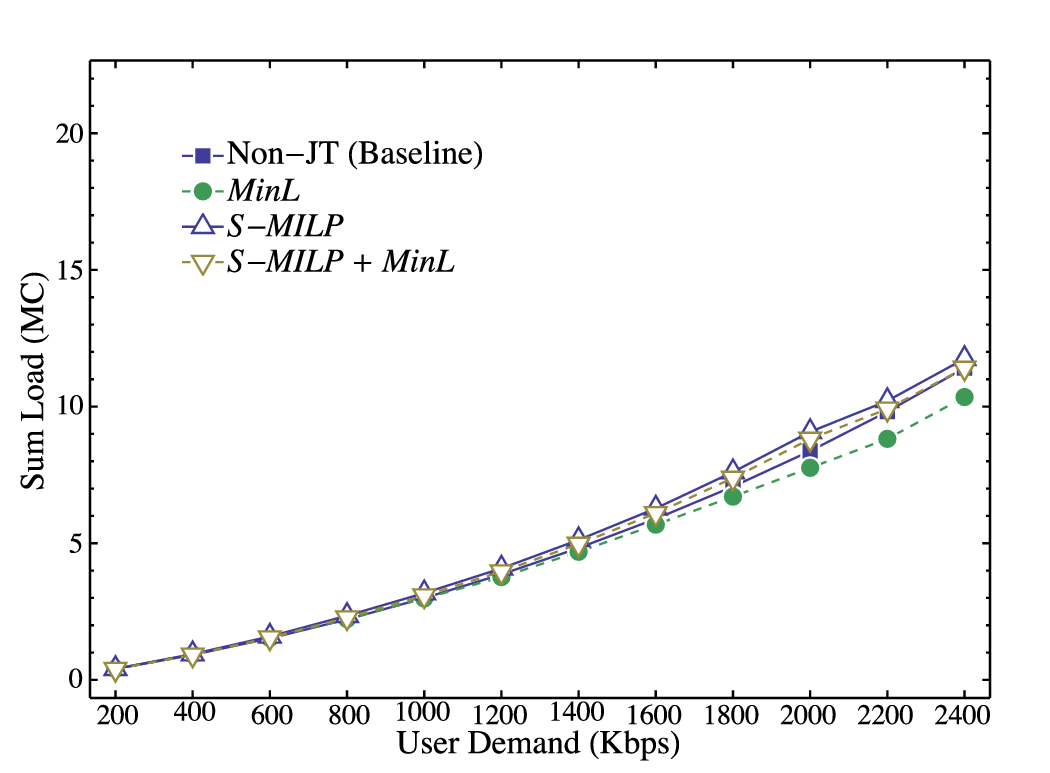}} \\
\subfigure[Sum Load (SC)]{\includegraphics[width=\linewidth]{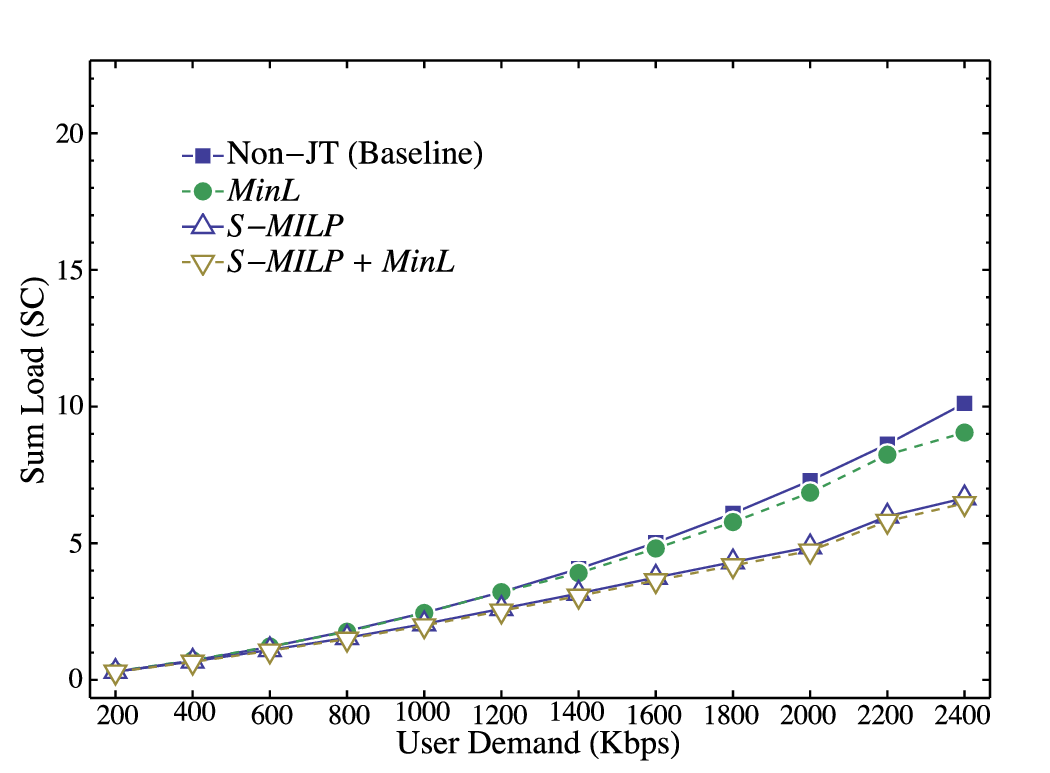}} 
\caption{Performance of sum load vs.\ user demand. 
} 
\label{fig:MinSumL}
\end{figure}

In \figurename~\ref{fig:MinSumL}, we show the performance of the algorithms for
\MSL{}. The main results of \textit{MinSumL} are given in
\figurename~\ref{fig:MinSumL}(a). 
For low UE demand, all cells are lightly loaded in the baseline solution of
Non-JT\@, and, from the lower bound, one concludes that 
the amount of possible improvement is negligible.
The impact of optimization becomes much apparent and very significant with high
user demand.  
For the highest user demand in the figure, 
\textit{MinL} reduces the sum load by $10\%$, compared to the baseline, and the
improvements by using \textit{S--MILP} and \textit{S--MILP} + \textit{MinL}, are
$15\%$ and $17.5\%$, respectively. Note that even for global optimum, the improvement 
can be no more than $18\%$, as shown by the lower bound, which is on the
optimistic side of the global optimum. Thus, from \textit{MinL},
\textit{S--MILP}, to \textit{S--MILP} + \textit{MinL}, the performance
progressively improve, and the latter is near optimal. 
We remark that the observation applies also if the results are averaged over the
demand range.

To gain further insights of the results in \figurename~\ref{fig:MinSumL}(a), we show
the sum load of MCs and SCs separately in \figurename~\ref{fig:MinSumL}(b) and
(c), respectively. These two figures reveal that, in comparison to the baseline
solution, the load of SCs due to optimization is significantly reduced, whereas
the optimized MC load increases slightly.  One can infer that, in order to
minimize the sum load, MCs expand, by using JT to serve some UEs of which the
home cells are SCs. This is because the
load of MC and SC are equally important in the objective function in
question, but the number of SCs are twice of MCs.
Thus, in optimizing the sum load, the MCs tend to have higher load than the
baseline solution. This insight explains why the \textit{S--MILP} performs better
than \textit{MinL} in \figurename~\ref{fig:MinSumL}(a), as the latter is
designed based on the partial optimality conditions for monotonically reducing
cell load.  On the other hand, \textit{MinL} is able to complement
\textit{S--MILP}, as demonstrated by the fact that the combined use of the two
achieves the best performance.

\subsection{Performance Comparison for MinMaxL}
\label{sec:resminmaxl}

\begin{figure}[h!] 
\centering 
\subfigure[Max Load (MC+SC)]{\includegraphics[width=\linewidth]{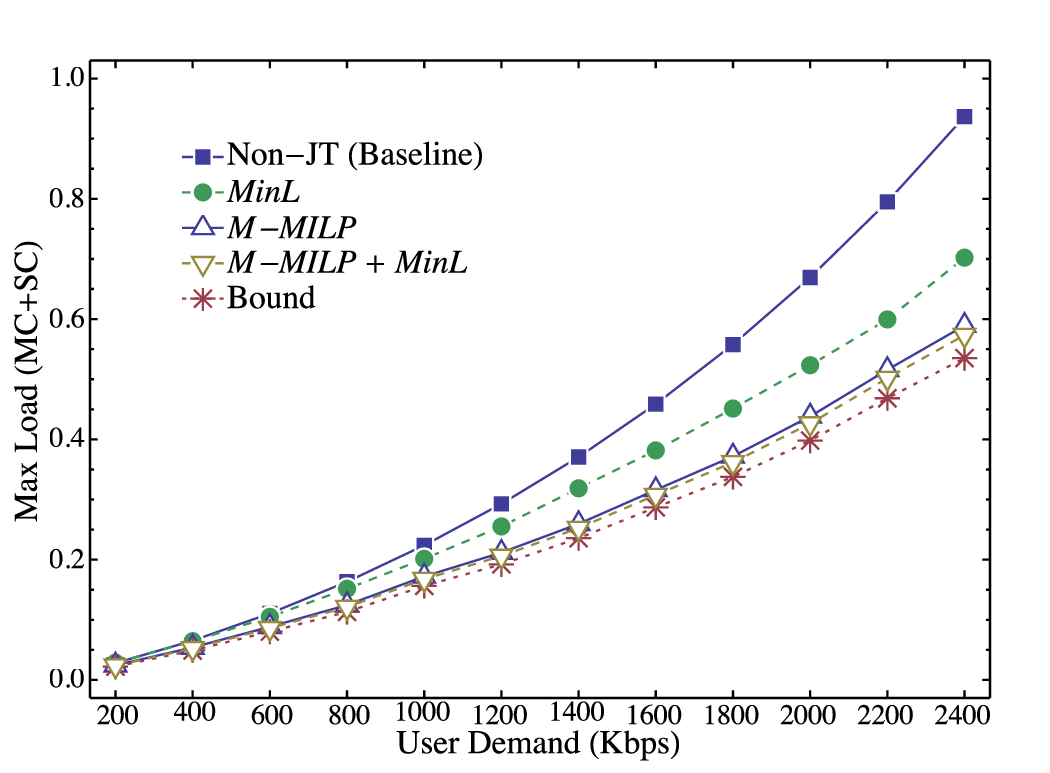}} \\
\subfigure[Max Load (MC)]{\includegraphics[width=\linewidth]{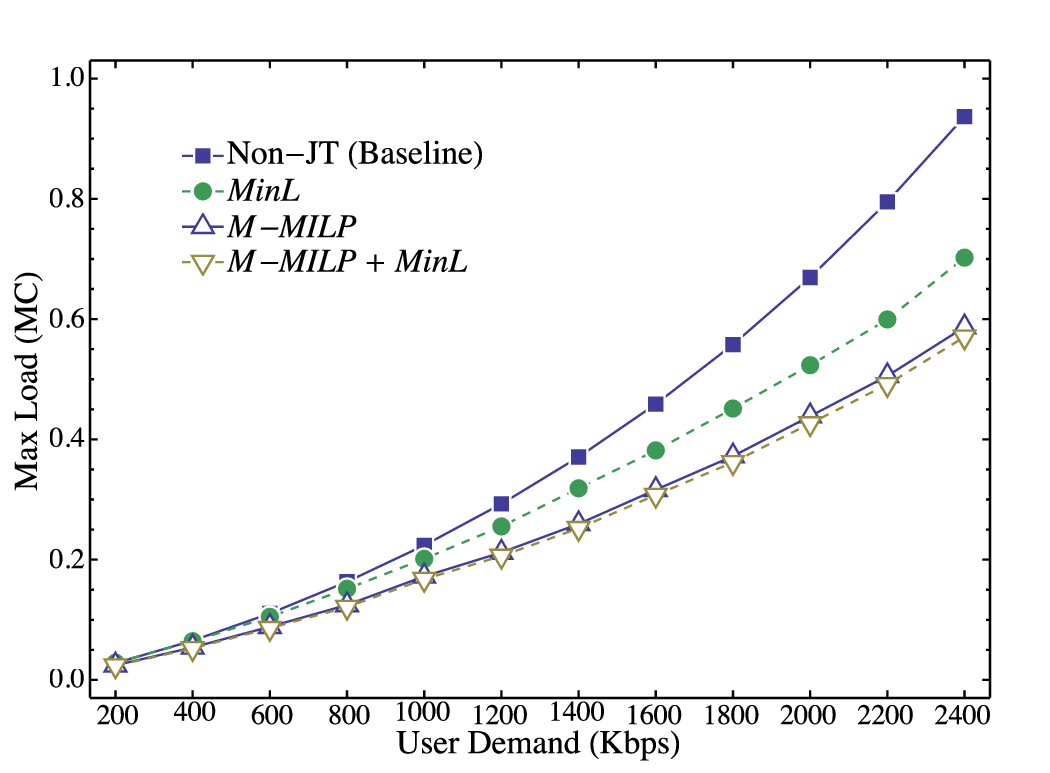}} \\ 
\subfigure[Max Load (SC)]{\includegraphics[width=\linewidth]{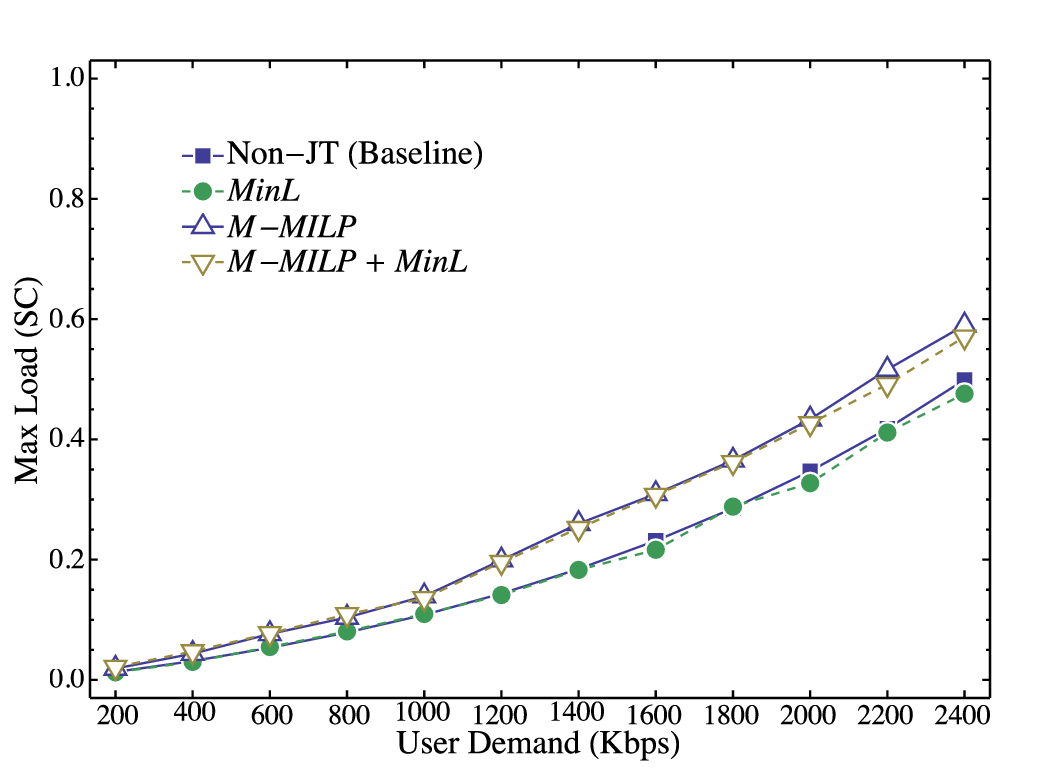}} 
\caption{Performance of max load vs.\ user demand. 
}
\label{fig:MinMaxL} 
\end{figure}

The performance of the algorithms for \MML{} is shown in 
\figurename~\ref{fig:MinMaxL}. The results are structured in three subfigures, 
following the construction of \figurename~\ref{fig:MinSumL}. 
From \figurename~\ref{fig:MinMaxL}(a), for \MML{} the performance gain by the
proposed algorithms is considerably higher than \MSL{}.
For the highest user demand under consideration, 
\textit{MinL} reduces the maximum load by almost $23\%$, compared to the
baseline, and the improvements via \textit{M--MILP} and \textit{M--MILP} +
\textit{MinL} grow to $34\%$ and $37\%$, respectively. Moreover, similar to
\figurename~\ref{fig:MinSumL}(a), \textit{M--MILP} +
\textit{MinL} performs closely to the lower bound, reinforcing the observation
of its near optimality. That is, there is little room left for further
performance improvement.   

We argue that MCs benefit considerably more than SCs from JT in \MML{}.  
By inspecting \figurename~\ref{fig:MinMaxL}(b) and (c), it is apparent that 
algorithm \textit{MinL} reduces the maximum load mainly for the MCs. For the SCs, the 
reduction of maximum load by algorithm \textit{MinL} is almost negligible. From the
results of the two MILP-based solutions, one discovers that the overall maximum
load should be improved at the cost of allowing higher SC load than the baseline.
In HetNets, MCs typically stand for higher load than SCs by the default
association strategy based on received power (i.e., our baseline solution).
Thus reducing the maximum load has to rely on the offloading effect of SCs, and
this effect is achieved via participation of SCs in JT.

The results in \figurename~\ref{fig:MinMaxL} also illustrate the
trade-off between optimality and complexity. Using the low-complexity
\textit{MinL} (Algorithm~\ref{alg:MinL}), the performance improvement amounts to
23\%, whereas the improvement grows to 34\% with \textit{M--MILP}
(Algorithm~\ref{alg:S--MILP}), with the price that the latter has high
complexity due to the use of MILP.

In general, the true optimality gaps of \textit{S--MILP} and
\textit{M--MILP} (or any sub-optimal algorithm) are difficult to
obtain, because doing so requires the value of the global optimum that
is not known. However, by selecting the end points as detailed in
Section~\ref{sec:bounds}, \textit{S--MILP} and \textit{M--MILP} are
also effective bounding schemes for performance evaluation, because
the gap to the optimum cannot exceed the gap in relation to the bound.
From the empirical results in \figurename~\ref{fig:MinSumL}(a) and
\figurename~\ref{fig:MinMaxL}(a), \textit{S--MILP} and \textit{M--MILP} 
perform well in delivering a close-to-optimal solution, because the
average gap to the bounds is small, with an average of merely 3\%.

Recall that an alternative linear approximation for both
\textit{S--MILP} and \textit{M--MILP} is to use the first-order Taylor
series. We have performed additional simulations with this
approximation, in which the derivative is taken at the middle point of
the interference interval considered.  The results show that this
linear approximation has almost identical performance as the linear
approximation discussed in Section~\ref{sec:linear}.

\subsection{Algorithm Scalability}

In \figurename~\ref{fig:large_scale}, we show additional results for
larger-scale scenarios via densification.  Specifically, both the
numbers of UEs and SCs are doubled. For these scenarios, we compare
our scalable algorithm \textit{MinL} with the baseline solution of
Non-JT. For sum load, we consider normalized values with respect to the number of cells, 
i.e., the average load,
because the number of cells in the figure differs by curve. Overall,
the results for both \textit{MinSumL} and
\textit{MinMaxL} show the same trends as those exhibited in
Sections \ref{sec:ressuml} and \ref{sec:resminmaxl}.  There are
however some additional insights.  For average load, the optimization
algorithm yields obvious improvements. When the number of SCs is
doubled, however, the improvement is insignificant. This is because,
with a dense network deployment, a UE is very likely to have very good
link to its home cell, and hence JT does not exhibit much benefit in average.
Interestingly, for reducing maximum load, the results are quite the
opposite: The improvement due to optimization is much more significant
for dense cell deployment. The reason is that additional cells
increase the dimension of freedom in optimizing JT to offload the most
loaded cell, whereas for a sparse network, each cell acts as the home cell
of quite many UEs for which there are no other cell having comparable
link quality and hence there is little room for offloading via JT.
Therefore, our optimization algorithm is useful for both sparse and
dense scenarios for minimizing the average and maximum cell load,
respectively.

\begin{figure}[h!] 
\centering 
\subfigure[Averaged Sum
Load]{\includegraphics[width=\linewidth]{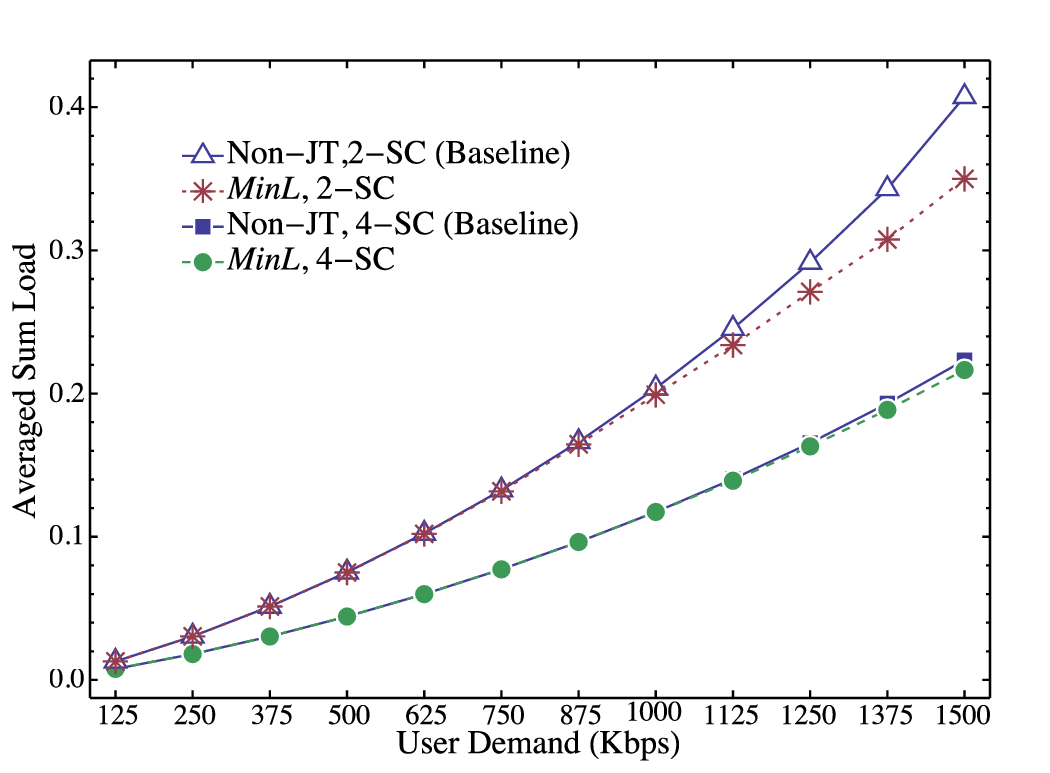}} \\
\subfigure[Maximum Load]{\includegraphics[width=\linewidth]{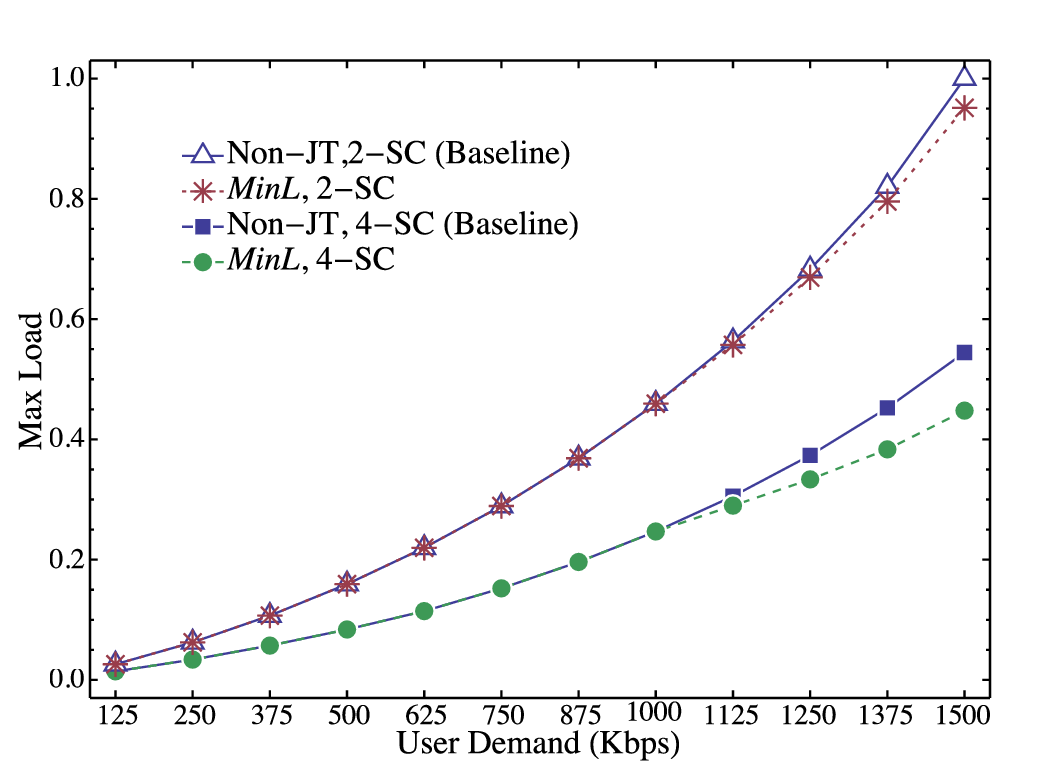}}  
\caption{Numerical results with large network scale (60 UEs, 2 or 4 SCs in each
    hexagonal region).}
\label{fig:large_scale} 
\end{figure}

In Section~\ref{sec:system_model}, we commented on the trade-off
between accuracy and complexity in interference modeling.  In order to
examine the accuracy in the modeling approach, we have performed
post-processing to the solutions obtained for all the network
scenarios and instances, using the more exact interference expression.
The results from post-processing exhibits an average deviation of 
$0.7\%$ in the performance values (cell load levels).
Hence the modeling approach used is indeed
reasonable.

\section*{Acknowledgement}
This work has been supported by the Swedish Research Council and the
Link\"{o}ping-Lund Excellence Center in Information Technology (ELLIIT), Sweden,
and the European Union Marie Curie project MESH-WISE (FP7-PEOPLE-2012-IAPP\@:
324515), DECADE (H2020-MSCA-2014-RISE\@: 645705), and WINDOW
(FP7-MSCA-2012-RISE\@: 318992). The work of D. Yuan has been
carried out within European FP7 Marie Curie IOF project 329313.

\section{Conclusion} 
\label{sec:conclusion}

The paper has investigated how to optimize cell-UE association with JT for
load-coupled LTE HetNets. The cell load refers to the amount of time-frequency
resource for serving user demand, and the cells influence each others' load
levels due to mutual interference. Two optimization problems, \MSL{} and \MML{},
have been formulated and analyzed. Solutions to the two problems strive for the
overall resource efficiency and load balancing among cells, respectively. We
have thoroughly proved the $\mathcal{NP}$-hardness of the two problems. A linearization
scheme has been developed, which yields an MILP-based solution approach to the
two problems. The approach also enables to compute lower bounds of global
optimum. Moreover, several optimality conditions have been derived, leading to a
second solution approach. Simulation results show the benefits of JT for
improving resource efficiency and load balancing, and the improvement is
significant for high user demand. The results also demonstrate that a combined
use of the two approaches performs closely to global optimum.




\begin{thebibliography}{10}

\bibitem{Andrews:ez} J.~G. Andrews, S.~Buzzi, W.~Choi, S.~V. Hanly, A.~Lozano,
A.~C.~K. Soong, and J.~C.  Zhang,
\newblock{``What will 5G be?''} \newblock\textit{IEEE Journal on Selected Areas
in Communications}, 32(6):
1065--1082, 2014.

\bibitem{Porcello:2014wq} V.~Jungnickel, K.~Manolakis, W.~Zirwas, B.~Panzner,
V.~Braun, M.~Lossow, M.~Sternad,
R.~Apelfr\"{o}jd, and T.~Svensson,  \newblock{``The role of small cells,
coordinated multipoint, and massive MIMO in 5G''}, 
\newblock\textit{IEEE Communications Magazine}, 52(5): 44--51,
2014.

\bibitem{Siomina:eq} I.~Siomina and D.~Yuan,  \newblock{``Analysis of cell load
coupling for LTE network
planning and optimization''},  \newblock\textit{IEEE Transactions on Wireless
Communications}, 11(6):
2287--2297, 2012.

\bibitem{Siomina:2009bp} I.~Siomina, A.~Furusk\"{a}r, and G.~Fodor. 
\newblock{``A mathematical framework for
statistical QoS and capacity studies in OFDM networks''}, 
\newblock\textit{IEEE International Symposium on
Personal, Indoor and Mobile Radio Communications (PIMRC)}, 2009.

\bibitem{Cavalcante:2014jd2} R. L. G. Cavalcante, S. Stanczak, M. Schubert, A.
Eisenblaetter, and U. Tuerke,
``Toward energy-efficient 5G wireless communications technologies: tools for
decoupling the scaling of
networks from the growth of operating power'', \textit{IEEE Signal Processing
Magazine}, 31(6): 24--34, 2014.

\bibitem{IViering:2009tq}	A. J. Fehske, I. Viering, J. Voigt, C. Sartori, S.
Redana, and G. P. Fettweis, ``A
mathematical perspective of self-optimizing wireless networks'',
\textit{IEEE International Conference on Communications (ICC)}, 2009.


\bibitem{Fehske:2013gn} A. J. Fehske, H. Klessig, J. Voigt, and G. P. Fettweis,
``Concurrent load-aware
adjustment of user association and antenna tilts in self-organizing radio
networks'', \textit{IEEE
Transactions on Vehicular Technology}, 62(5): 1974–-1988, 2013.

\bibitem{Fehske:2012iw} A.~J.~Fehske and G.~P.~Fettweis, 
\newblock{``Aggregation of variables in load models
for interference-coupled cellular data networks''},  \newblock\textit{IEEE
International Conference on
Communications (ICC)}, 2012.

\bibitem{Cavalcante:2014jd} R.~L.~G.~Cavalcante, S.~Sta\'{n}czak, M.~Schubert,
A.~Eisenblaetter, and
U.~T\"{u}rke,  \newblock{``Toward energy-efficient 5G wireless communications
technologies: Tools for
decoupling the scaling of networks from the growth of operating power''}, 
\newblock\textit{IEEE Signal
Processing Magazine}, 31(6):24--34, 2014.



\bibitem{Siomina:2013ew} I.~Siomina and D.~Yuan,  \newblock{``Optimization
approaches for planning small cell
locations in load-coupled heterogeneous LTE networks''},  \newblock\textit{IEEE
International Symposium on
Personal Indoor and Mobile Radio Communications (PIMRC)}, 2013.

\bibitem{Siomina:2014be} I.~Siomina and D.~Yuan,  \newblock{``Optimizing small
cell range in heterogeneous and
load-coupled LTE networks''},  \newblock\textit{IEEE Transactions on Vehicular
Technology}, 64(5):
2169--2174, 2014.

\bibitem{You:2015wva} L.~You, L.~Lei, and D.~Yuan,  \newblock{``Load balancing
via joint transmission in
heterogeneous LTE\@: Modeling and computation''},  \newblock\textit{IEEE
International Symposium on Personal
Indoor and Mobile Radio Communications (PIMRC)}, 2015.


\bibitem{Athanasiou:2015hb} G. Athanasiou, P. C. Weeraddana, C. Fischione, and
L. Tassiulas, ``Optimizing client
association for load balancing and fairness in millimeter-wave wireless
networks'', \textit{IEEE/ACM
Transactions on Networking}, 23(3): 836--850, 2015.

\bibitem{Sun:2015kn} R. Sun, M. Hong, and Z. Luo, ``Joint downlink base station
association and power control
for max-min fairness: computation and complexity'', \textit{IEEE Journal on
Selected Areas of
Communications}, 33(6): 1040--1054, 2015.

\bibitem{Anonymous:M4eIBsQg} A. G. Gotsis, S. Stefanatos, and A. Alexiou,
“Optimal user association for massive
MIMO empowered ultra-dense wireless networks”, \textit{International Conference
on Communication Workshop
(ICCW)}, 2015

\bibitem{Singh:2014bs} S. Singh, F. Baccelli, and J. G. Andrews, ``On
association cells in random heterogeneous
networks'', \textit{IEEE Wireless Communications Letters}, 3(1):70--73, 2014.

\bibitem{Li:2013ho} Q. Li, R. Q. Hu, Y. Qian, and G. Wu, ``Intracell
cooperation and resource allocation in a
heterogeneous network with relays'', \textit{IEEE Transactions on Vehicular
Technology}, 62(4): 1770--1784,
2013.


\bibitem{Hong:2012ky} M. Hong and A. Garcia, ``Mechanism design for base
station association and resource
allocation in downlink OFDMA network'', \textit{IEEE Journal on Selected Areas
of Communications}, 30(11):
2238–2250, 2012.

\bibitem{Ye:2013ha} Q.~Ye, B.~Rong, Y.~Chen, M.~Al-Shalash, C.~Caramanis, and
J.~G.~Andrews,  \newblock{``User
association for load balancing in heterogeneous cellular networks''}. 
\newblock\textit{IEEE Transactions on
Wireless Communications}, 12(6): 2706--2716, 2013.

\bibitem{Nigam:2014cd} G.~Nigam, P.~Minero, and M.~Haenggi. 
\newblock{``Coordinated multipoint joint
transmission in heterogeneous networks''},  \newblock\textit{IEEE Transactions
on Communications}, 62(11):
4134--4146, 2014.

\bibitem{Lakshmana:2014eu} T.~R.~Lakshmana and B.~Makki,  \newblock{``Frequency
allocation in non-coherent joint
transmission CoMP networks''},  \newblock\textit{IEEE International Conference
on Communications (ICC)},
2014.

\bibitem{TaSiAnJo14} R.~Tanbourgi, S.~Singh, J.~G.~Andrews, and F.~K.~Jondral, 
\newblock{``A tractable model for noncoherent joint-transmission base station
cooperation}'', \newblock\textit{IEEE Transactions on Wireless
Communications}, 13(9): 4959--4973, 2014.

\bibitem{BaGi15} F.~Baccelli and A.~Giovanidis, \newblock{``A stochastic geometry
framework for analyzing pairwise-cooperative cellular networks''}, 
\newblock\textit{IEEE Transactions on Wireless Communications}, 14(2):794--808, 2015.

\bibitem{Schubert:2014ud} M.~Schubert and H.~Boche,  \newblock\textit{A General
Framework for Interference
Management and Network Utility Optimization}, Springer, 2014.

\bibitem{Anonymous:Dv6oz3oX} S.~Boyd and L.~Vandenberghe, \textit{Convex
Optimization}, Cambridge, 2014.

\bibitem{Yates:1995eh}
R. D. Yates,
``A framework for uplink power control in cellular radio systems''. 
\textit{IEEE Journal on Selected Areas in Communications}, 13(7): 1341--1347,
1995.

\bibitem{Huang:1998uc} 
    C. Y. Huang and R. D. Yates, ``Rate of convergence for minimum power
    assignment algorithms in cellular radio systems'', \textit{Wireless
    Networks}, 4(4):223--231, 1998.

\bibitem{3gpp} \textit{3GPP TR 36.913}, V13.0.0, 2016.
\bibitem{Anonymous:GtzCaaVU} \textit{3GPP TR 36.814}, V9.0.0, 2010.



\end{thebibliography}
\end{document}